\documentclass[preprint]{imsart}

\RequirePackage[OT1]{fontenc}
\RequirePackage{amsthm,amsmath}
\RequirePackage[numbers]{natbib}
\RequirePackage[colorlinks,citecolor=blue,urlcolor=blue]{hyperref}
\usepackage{amsmath,amssymb,amsfonts,amsthm,enumerate,natbib,color,ifthen,hyperref}
\usepackage{xargs}
\usepackage[textwidth=4cm, textsize=footnotesize]{todonotes}
\usepackage[latin1]{inputenc}

\usepackage{float}
\usepackage[caption = false]{subfig}
\usepackage{graphicx}

\usepackage{aliascnt,bbm}
\def\rset{\mathbb R}
\def\zset{\mathbb Z}

\def\eqsp{\;}

\newcommand{\pscal}[2]{\left\langle#1,#2\right\rangle}

\newcommand{\eqdef}{\ensuremath{\stackrel{\mathrm{def}}{=}}}

\def\Xset{\mathcal{X}} 
\def\B{\mathcal{B}} 
\def\cB{\mathsf{B}} 
\def\e{\mathcal{E}}

\def\D{\mathcal{D}}
\def\A{\mathcal{A}}

\newcommandx\sequence[3][2=t,3=\zset]
{\ifthenelse{\equal{#3}{}}{\ensuremath{\{ #1_{#2}\}}}{\ensuremath{\{ #1_{#2}, \eqsp #2 \in #3 \}}}}

\def\PP{\mathbb{P}} 
\newcommand{\CPP}[3][]
{\ifthenelse{\equal{#1}{}}{{\mathbb P}\left(\left. #2 \, \right| #3 \right)}{{\mathbb P}_{#1}\left(\left. #2 \, \right | #3 \right)}}
\def\PE{\mathbb{E}} 
\newcommand{\CPE}[3][]
{\ifthenelse{\equal{#1}{}}{{\mathbb E}\left[\left. #2 \, \right| #3 \right]}{{\mathbb E}_{#1}\left[\left. #2 \, \right | #3 \right]}}

\def\tv{\mathrm{tv}}

\def\Cset{\mathcal{C}} 

\def\I{\textsf{I}}

\usepackage{bm}

\theoremstyle{plain}
\newtheorem{theorem}{Theorem}
\newtheorem{assumption}{H\hspace{-3pt}}

\newaliascnt{proposition}{theorem}

\aliascntresetthe{proposition}

\newaliascnt{lemma}{theorem}
\newtheorem{lemma}[lemma]{Lemma}
\aliascntresetthe{lemma}
\newaliascnt{corollary}{theorem}

\aliascntresetthe{corollary}

\theoremstyle{definition}
\newaliascnt{definition}{theorem}

\aliascntresetthe{definition}

\newtheorem{algorithm}{Algorithm}

\newaliascnt{remark}{theorem}
\newtheorem{remark}[remark]{Remark}
\aliascntresetthe{remark}

\newaliascnt{example}{theorem}

\aliascntresetthe{example}

\def\rmd{\mathrm{d}}

\def\1{\mathbbm{1}}

\startlocaldefs
\numberwithin{equation}{section}
\theoremstyle{plain}

\endlocaldefs

\begin{document}

\begin{frontmatter}
\title{Approximate spectral gaps for Markov chains mixing times in high dimensions\thanksref{T1}}
\runtitle{Approximate spectral gaps for Markov chains}
\thankstext{T1}{This work is partially supported by the NSF grant DMS1513040.}

\begin{aug}
\author{\fnms{Yves} \snm{Atchad\'e}\ead[label=e1]{atchade@bu.edu}}


\affiliation{Boston University\thanksmark{m1}}

\address{111 Cummington Mall, Boston, 02215, MA, United States\\
\printead{e1}
}
\end{aug}

\begin{abstract}
This paper introduces a concept of approximate spectral gap to analyze the mixing time of Markov Chain Monte Carlo (MCMC) algorithms for which the usual spectral gap is degenerate or almost degenerate.  We use the idea to analyze a class of MCMC algorithms to sample from mixtures of densities.  As an application we study the mixing time of a Gibbs sampler for variable selection in linear regression models. Under some regularity conditions on the signal and the design matrix of the regression problem, we show that for well-chosen  initial distributions the mixing time of the Gibbs sampler is polynomial in the dimension of the space. 
\end{abstract}



\end{frontmatter}

\section{Introduction}\label{sec:intro}
Understanding the type of problems for which fast Markov Chain Monte Carlo (MCMC) sampling is possible is a question of fundamental interest. The study of the size of the spectral gap is a widely used approach to gain insight into the behavior of MCMC algorithms. However this technique may be inapropriate when dealing with distributions with small isolated local modes. 
To be more precise, let $\pi$ be some probability measure of interest on some measure space $\Xset$, and let  $K$ be a Markov kernel with invariant distribution $\pi$. For the purpose of sampling from $\pi$ using $K$, one can represent an isolated local mode (to which $K$ is sensitive) as a subset $A$ such that $K(x,\Xset\setminus A)$ is small compared to $\pi(\Xset\setminus A)$ for all $x\in A$.  In this  case, $K$  will have a small conductance (see (\ref{PhiK}) for definition), and hence a small spectral gap.  Note however that if $\pi(A)$ is also small (that is we are dealing with a small isolated mode $A$), then, since
\[\int_{\Xset\setminus A}\pi(\rmd x)K(x,A) = \int_{A}\pi(\rmd x)K(x,\Xset\setminus A),\]
 we see that the set $A$ will be typically hard to reach in the first place. Hence, any finite Markov chain $\{X_0,\ldots,X_n\}$ say, with transition kernel $K$ and initialized in $\Xset\setminus A$ is unlikely to visit $A$. Yet,  for $n$ large, $X_n$ may still be a good approximate sample from $\pi$,  since $\pi(A)$ is small. This implies that the mixing time predicted by the standard spectral  gap may markly differ from the practical behavior of these finite chains.  Motivated by this problem, and building on the $s$-conductance of L. Lovasz and M. Simonovits (\cite{lovacz:simonovits:90,lovasz:simonovits93}), we develop an idea of approximate spectral gap (that we call $\zeta$-spectral gap, for some $\zeta\in [0,1)$) which allows us to measure the mixing time of a Markov chain while discounting the ill-effect of overly small (and potentially problematic) sets.  

Mixtures are good examples of probability distributions with isolated local modes. We use the idea to analyze a class of MCMC algorithms to sample from mixtures of densities.  Much is known on the computational complexity of various MCMC algorithms for log-concave densities (see e.g. \cite{lovacz:simonovits:90,lovasz:simonovits93,friezeetal94,lovasz99,lovasz:vempala07,dwivedi:etal:18} and the references therein). However these results cannot be directly applied to mixtures, since  a mixture of log-concave densities is not log-concave in general.  By augmenting the variable of interest to include the mixing variable, a Gibbs sampler  can be used to sample from a mixture. A very nice lower bound on the spectral gap of such Gibbs samplers (and generalizations thereof) is developed in \cite{madras:randall:02}. However the analysis of \cite{madras:randall:02} typically leads to  mixing times that grow exponentially fast with the dimension of the space. We re-examine \cite{madras:randall:02}'s argument using the concept of $\zeta$-spectral gap, leading to Theorem \ref{thm:mixing} that gives potentially better dependence on the dimension.

Our initial motivation into this work is in large-scale Bayesian variable selection problems. The Bayesian posterior distributions that arise from these problems are typically mixtures of log-concave densities with very large numbers of components, and the aforementioned Gibbs sampler is commonly used  for sampling (see e.g. \cite{george:mcculloch97,narisetti:he:14}). 
We show that the proposed  concept of $\zeta$-spectral gap and Theorem \ref{thm:mixing} can be combined with Bayesian posterior contraction principles to  show that the algorithm -- with a good initialization -- has a mixing time that is polynomial in the number of regressors in the model (see Theorem \ref{thm:mixing:lm}).

The paper is organized as follows. We develop the concept of $\zeta$-spectral gap in Section \ref{sec:approx:sp}. The main result there is Lemma \ref{lem:key}. In Section \ref{sec:mixing} we study the mixing time of mixtures of Markov kernels, and derive (Theorem \ref{thm:mixing})  a generalization of Theorem 1.2 of \cite{madras:randall:02}.  We put these two results together to analysis the linear regression model in Section \ref{sec:ex:linregr}, leading to Theorem \ref{thm:mixing:lm}.  Some numerical simulations are detailed in Section \ref{sec:sim}.

\section{Approximate spectral gaps for Markov chains}\label{sec:approx:sp}
 Let $\pi$ be a probability measure on some Polish space $(\Xset,\B)$ (where $\B$ is its Borel sigma-algebra), equipped with a reference sigma-finite measure denoted $\rmd x$. In the applications that we have in mind, $\Xset$ is the Euclidean space $\rset^p$ equipped with its Lebesgue measure. We assume that $\pi$ is absolutely continuous with respect to $\rmd x$, and we will abuse notation and use $\pi$ to denote both $\pi$ and its density: $\pi(\rmd x) = \pi(x)\rmd x$. We let $L^2(\pi)$ denote the Hilbert space of all real-valued square-integrable (wrt $\pi$) functions on $\Xset$, equipped with the inner $\pscal{f}{g}_\pi\eqdef\int_\Xset f(x)g(x)\pi(\rmd x)$ with associated norm $\|\cdot\|_{2,\pi}$. More generally, for $s\geq 1$, we set $\|f\|_{s,\pi}\eqdef\left(\int_\Xset |f(x)|^s\pi(\rmd x)\right)^{1/s}$. For $s=+\infty$, $\|f\|_{s,\pi}$ is defined as the essential supremum of $|f|$ with respect to $\pi$. If $P$ is a Markov kernel on $\Xset$, and $n\geq 1$ an integer, $P^n$ denotes the $n$-th iterate of $P$, defined recursively as $P^n(x,A) \eqdef \int_\Xset P^{n-1}(x,\rmd z) P(z,A)$, $x\in\Xset$, $A$ measurable. If $f:\;\Xset\to\rset$ is a measurable function, then $Pf:\;\Xset\to\rset$ is the function defined as $Pf(x) \eqdef \int_\Xset P(x,\rmd z)f(z)$, $x\in\Xset$, assuming that the integral is well defined. And if $\mu$ is a probability measure on $\Xset$, then $\mu P$ is the probability on $\Xset$ defined as $\mu P(A) \eqdef\int_\Xset \mu(\rmd z) P(z,A)$, $A\in\B$. The total variation distance  between two probability measures $\mu,\nu$ is defined as
 \[\|\mu-\nu\|_\tv \eqdef 2 \sup_{A \in \B} \left(\mu(A)-\nu(A)\right).\]
  
Let $K$ be a Markov kernel on $\Xset$ that is reversible with respect to $\pi$. That is for all $A,B\in\B$,
\[\int_A\pi(\rmd x)\int_B K(x,\rmd y) =  \int_B\pi(\rmd x)\int_A K(x,\rmd y).\]
We will also assume throughout that $K$ is lazy in the sense that $K(x,\{x\}) \geq \frac{1}{2}$. The concept of spectral gap and the related Poincare's inequalities are  commonly used to quantify Markov chains mixing times. For $f\in L^2(\pi)$, we set $\pi(f)\eqdef \int_\Xset f(x)\pi(\rmd x)$, $\textsf{Var}_\pi(f) \eqdef \|f-\pi(f)\|_{2,\pi}^2$,  and $\e(f,f)\eqdef \frac{1}{2}\int\int (f(y)-f(x))^2\pi(\rmd x) K(x,\rmd y)$.  The spectral gap of $K$ is  then defined as 
\[\textsf{SpecGap}(K) \eqdef \inf\left\{\frac{\e(f,f)}{\textsf{Var}_\pi(f)},\;f\in L^2(\pi),\;\mbox{ s.t. }\; \textsf{Var}_\pi(f)>0\right\}.\]
It is well-known  and easy  to establish (see for instance \cite{montenegro:etal:06}~Corollary 2.15) that if $\pi_0(\rmd x) = f_0(x)\pi(\rmd x)$, and $f_0\in L^2(\pi)$, then
\begin{equation}\label{eq:spec:gap:cv}
\|\pi_0 K^n-\pi\|_\tv^2  \leq \textsf{Var}_\pi(f_0)\left(1-\textsf{SpecGap}(K)\right)^{n}.\end{equation}
Therefore, lower-bounds on the spectral gap can be used to derive upper-bounds on the mixing time of $K$. We refer the reader to (\cite{sinclair:jerrum89,sinclair:90,diaconis:stroock:91,montenegro:etal:06}) for more details, and for various strategies to lower-bound $\textsf{SpecGap}(K)$. In many examples, the conductance of $K$,  defined as  
\begin{equation}\label{PhiK}
\Phi(K) \eqdef \inf\left\{\frac{\int_A\pi(\rmd x) K(x,A^c)}{\pi(A)\pi(A^c)},\;A\in\B:\; 0<\pi(A)<1\right\}, \end{equation}
is easier  to control than the spectral gap. Cheeger's inequality for Markov chains (\cite{lawler:sokal:88,sinclair:jerrum89}) can then be used to translate a lower-bound on $\Phi(K)$ into a lower-bound on the spectral gap:
\begin{equation}\label{cond:sg}
\frac{1}{8}\Phi(K)^2\leq \textsf{SpecGap}(K) \leq \Phi(K).\end{equation}

The concept of $s$-conductance introduced by  L. Lovacz and M. Simonivits (\cite{lovacz:simonovits:90,lovasz:simonovits93}, see also \cite{lovasz:vempala07}) as a generalization of the conductance has proven very useful. For $\zeta\in [0,1/2)$ -- using a definition slightly different from \cite{lovacz:simonovits:90,lovasz:simonovits93} -- we define the $\zeta$-conductance of the Markov kernel $K$  as 
\[\Phi_\zeta(K)\eqdef \inf\left\{\frac{\int_A\pi(\rmd x)K(x,A^c)}{(\pi(A)-\zeta)(\pi(A^c)-\zeta)},\;\;\zeta<\pi(A) < \frac{1}{2}\right\},\]
where the infimum above is taken over measurable subsets of $\Xset$. Note that $\Phi_0(K)= \Phi(K)$. Plainly put, $\Phi_\zeta(K)$ captures the same concept as $\Phi(K)$, except that in $\Phi_\zeta(K)$ we disregard sets that are either too small or too large under $\pi$.  It turns out that $\Phi_\zeta(K)$ still controls the mixing time of $K$ up to an additive constant that depends on $\zeta$ (see \cite{lovasz:simonovits93}~Corollary 1.5). One important drawback of the $\zeta$-conductance is that the arguments that relate $\Phi_\zeta(K)$ to the mixing time of $K$ (Theorem 1.4 of \cite{lovasz:simonovits93}) is rather involved, and this  has limited the scope and the usefulness of  the concept. Furthermore there are many problems where direct bound on the spectral gap instead of the conductance is easier, and yields better results. This is for instance the case in discrete problems where canonical path arguments yields much sharper bounds on the Poincare constant (\cite{diaconis:stroock:91}).  

Motivated by the $\zeta$-conductance,  we introduce a similar concept of $\zeta$-spectral gap that directly approximates the spectral gap. And we show that the proposed $\zeta$-spectral gap still controls the mixing time  of the Markov chains.

Let $\|\cdot\|_\star:\;L^2(\pi)\to [0,\infty]$ denote a norm-like function on $L^2(\pi)$ with the following properties: $\|\alpha f\|_\star =|\alpha| \|f\|_\star$, if $\|f\|_\star =0$ then $\textsf{Var}_\pi(f)=0$, and 
\[\|K f\|_\star \leq \|f\|_\star,\;\;f\in L^2(\pi).\]
For $\zeta\in (0,1/2)$, we define the $\zeta$-spectral gap of $K$ as
\begin{equation}\label{eq:zeta:sg}
\textsf{SpecGap}_\zeta(K) \eqdef \inf\left\{\frac{\e(f,f)}{\textsf{Var}_\pi(f) - \frac{\zeta}{2}},\;\;f\in L^2(\pi),\; \textsf{Var}_\pi(f)> \zeta,\mbox{ and }\|f\|_{\star} = 1\right\}.\end{equation}
We note that $\textsf{SpecGap}_\zeta(K)$ depends on the choice of $\|\cdot\|_\star$, although we will not make that dependence explicit. We note also that if $\zeta=0$ and $\|f\|_\star=\|f\|_{2,\pi}$, then we recover $\textsf{SpecGap}_0(K) = \textsf{SpecGap}(K)$. Furthermore, given $f\in L^2(\pi)$, and writing $\bar f=f-\pi(f)$, we have
\[\frac{\e(f,f)}{\textsf{Var}_\pi(f) - \frac{\zeta}{2}} =\frac{\pi(\bar f^2) - \pscal{\bar f}{P\bar f}_\pi}{\pi(\bar f^2) - \frac{\zeta}{2}}.\]
 By the lazyness of the chain, $\pscal{\bar f}{P\bar f}_\pi\geq \pi(\bar f^2)/2$, and we deduce that $\textsf{SpecGap}_\zeta(K)$ is a quantity that always belongs to the interval $[0,1]$.
 

This idea of $\zeta$-spectral gap is somewhat similar to the concept of weak Poincare inequality developed for continuous-time Markov semigroups with zero spectral gap (\cite{liggett:91,rockner:wang:01,cattiaux:guillin:09}).  One key difference is that weak Poincare inequalities lead to sub-geometric rates of convergence of the semi-group, whereas the idea of $\zeta$-spectral gap as introduced here leads to a geometric convergence rate, plus an additive remainder that depends on $\zeta$. More precisely, we have the following analog of (\ref{eq:spec:gap:cv}). The proof is similar to the proof  of (\ref{eq:spec:gap:cv}), and is based on an argument from \cite{mihail:89}.

\begin{lemma}\label{lem:key}
Fix $\zeta\in (0,1/2)$.  Suppose that $\pi_0(\rmd x) = f_0(x)\pi(\rmd x)$ for a function $f_0\in L^2(\pi)$ such that $\|f_0\|_\star<\infty$. Then for all integer $n\geq 1$, we have
\[\|\pi_0 K^n -\pi\|_\tv^2  \leq  \max\left(\textsf{Var}_\pi(f_0),\zeta\|f_0\|_\star^2\right)\left(1 - \textsf{SpecGap}_\zeta(K)\right)^n + \zeta\|f_0\|_{\star}^2.\]
\end{lemma}
\begin{proof}
See Section \ref{sec:proof:lem:key}.
\end{proof}

We now highlight an approach to  lower bound $\textsf{SpecGap}_\zeta(K)$ and use Lemma \ref{lem:key}. This is the same approach used in the proof  of Theorem \ref{thm:mixing:lm}. Hence the following discussion can also be viewed as a rough sketch of the proof of Theorem \ref{thm:mixing:lm}. To proceed we first introduce a related concept  of restricted spectral gap. If $\Xset_0\subseteq\Xset$ is a non-empty measurable subset such that $\pi(\Xset_0)>0$, the $\Xset_0$-restricted spectral gap of $K$ is defined as
\[\textsf{SpecGap}_{\Xset_0}(K) \eqdef \inf\left\{\frac{\int_{\Xset_0}\int_{\Xset_0}\pi(\rmd x)K(x,\rmd y)(f(y)-f(x))^2}{\int_{\Xset_0}\int_{\Xset_0}\pi(\rmd x)\pi(\rmd y)(f(y)-f(x))^2},\;f:\;\Xset\to\rset\right\},\]
where the infimum is taken over all measurable functions $f$ such that \\$\int_{\Xset_0}\int_{\Xset_0}\pi(\rmd x)\pi(\rmd y)(f(y)-f(x))^2>0$.    The next result shows that these restricted spectral gaps can be used to lower bound $\textsf{SpecGap}_\zeta(K)$.
\begin{lemma}\label{lem:useful}
Given $\zeta\in (0,1/2)$, and taking $\|\cdot\|_\star=\|\cdot\|_{m,\pi}$, for some real number $m \in (2,+\infty]$, let $\Xset_\zeta$ be a measurable subset of $\Xset$ such that $\pi(\Xset_\zeta)\geq 1-\left(\frac{\zeta}{10}\right)^{1+ \frac{2}{m-2}}$. Then we have  
\[\textsf{SpecGap}_\zeta(K)\geq \textsf{SpecGap}_{\Xset_\zeta}(K).\]
\end{lemma}
\begin{proof}
See Section \ref{sec:proof:lem:useful}.
\end{proof}

We combine the above two lemmas as follows. Fix $\zeta_0\in(0,1)$. Suppose that we can choose the initial distribution $\pi_0$ such that $\|f_0\|_{m,\pi}\leq B$, for some constant $B\geq 1$ (warm start). In that  case Lemma \ref{lem:key} with $\|\cdot\|_\star = \|\cdot\|_{m,\pi}$, and $\zeta = \zeta_0^2/(2B^2)$ gives for all $n\geq 1$,
\[\|\pi_0 K^n -\pi\|_\tv^2  \leq  \max\left(1, \textsf{Var}_\pi(f_0)\right)\left(1 - \textsf{SpecGap}_\zeta(K)\right)^n + \frac{\zeta_0^2}{2}.\]
Therefore, for this given choice $\zeta = \zeta_0^2/(2B^2)$, if we can find  a set $\Xset_\zeta$ such that $\pi(\Xset_\zeta)\geq 1-\left(\frac{\zeta}{10}\right)^{1+ \frac{2}{m-2}}$, then Lemma \ref{lem:useful} asserts that $\textsf{SpecGap}_\zeta(K)\geq \textsf{SpecGap}_{\Xset_\zeta}(K)$. Hence it holds that
\[\|\pi_0 K^N -\pi\|_\tv  \leq \zeta_0,\;\; \mbox{ for all } \;\; N\geq \frac{\log\left(\frac{2B^2}{\zeta_0^2}\right)}{\textsf{SpecGap}_{\Xset_\zeta}(K)}.\]
If $\pi$ is a posterior distribution from some Bayesian analysis, posterior contraction results can be used to find the sets $\Xset_\zeta$. Furthermore,  standard techniques used to establish Poincare inequalities can be similarly applied to lower bound $\textsf{SpecGap}_{\Xset_\zeta}(K)$. We illustrate these ideas in Theorem \ref{thm:mixing:lm}.

\section{Mixing times of mixtures of Markov kernels}\label{sec:mixing}
We consider here the case where $\pi$ is a discrete mixture of log-concave densities of the form
\begin{equation}\label{pi:mix}
\pi(\rmd x) \propto \sum_{i\in\I} \pi(i,x)\rmd x,\end{equation}
for nonnegative measurable functions $\{\pi(i,\cdot),\;i\in\I\}$, where $\I$ is a nonempty finite set. 
Sampling from mixtures is more challenging than sampling from log-concave densities. For instance it is shown in  \cite{ge:etal:18}  that no polynomial-time MCMC algorithm exists to sample from mixtures of densities with inequal covariance matrix, if the algorithm uses only the marginal density of the mixture and its derivative.  One major shortcoming of  \cite{ge:etal:18} is that their algorithm is impractical when the number of mixture components is very large. In such settings, a Gibbs sampler is commonly employed (based on conditional distributions). We show below that this Gibbs sampler is fast mixing in some cases.

To avoid confusion we will write $\bar\pi$ to denote the joint distribution on $\I\times \Xset$ defined as
\[\bar\pi(D\times B) = \frac{\sum_{i\in D} \int_B\pi(i,x)\rmd x}{\sum_{i\in \I} \int_{\Xset}\pi(i,x)\rmd x},\;\;D\subseteq\I,\;B\in\B.\]
 Let $\pi(i\vert x)\propto \pi(i,x)$ (resp. $\pi(i)\propto \int_{\Xset}\pi(i,x)\rmd x$) denote the implied conditional (resp. marginal) distribution on $\I$, and let $\pi_i(\rmd x) \propto \pi(i,x)\rmd x$ be the implied conditional distribution on $\Xset$. For each $i\in\I$, let $K_i$ be a transition kernel on $\Xset$ with invariant distribution $\pi_i$. We assume that $K_i$ is reversible with respect to $\pi_i$, and ergodic (phi-irreducible and aperiodic).  We then consider  the Markov kernel $K$ defined as 
\begin{equation}\label{K:mix}
K(x,\rmd y) \eqdef \sum_{i\in\I}\pi(i\vert x) K_i(x,\rmd y),\end{equation}
that is reversible with respect to $\pi$  as in (\ref{pi:mix}).  In  \cite{madras:randall:02} the authors developed a very nice lower bound on the spectral gap of $K$ knowing the spectral gaps of the  $K_i$'s. 
Fix $\kappa>0$, and construct a graph on $\I$ such that there is an edge between $i,j\in\I$ if and only if
\[\int_\Xset\min\left(\pi_i(x),\pi_j(x)\right)\rmd x\geq \kappa.\]
If $D(\I)$ denotes the diameter of the graph thus defined\footnote{The diameter of a graph is the length (the number of edges) of the longest among all the shortest paths between all pairs of vertices.}, Theorem 1.2 of \cite{madras:randall:02} says that
\begin{equation}\label{bound:madras}
\textsf{SpecGap}(K) \geq \frac{\kappa}{2D(\I)}\min_{i\in\I} \left\{\pi(i)\textsf{SpecGap}(K_i)\right\}.\end{equation}
The lower bound in (\ref{bound:madras}) can be extremely small, particularly when $\I$ is large. Indeed, the ratio $\kappa / D(\I)$ would then be small: taking $\kappa$ large makes $D(\I)$ large.  Furthermore,  in problems where $\I$ is large,  $\pi(i)$ is typically exponentially small for many components $i$.   We have  the following analog of Lemma \ref{lem:useful}.
\begin{theorem}\label{thm:mixing}Let $\pi$ as in (\ref{pi:mix}), and $K$ as in (\ref{K:mix}) for some family $\{K_i,\;i\in\I\}$ of Markov kernels on $\Xset$. Choose $\|\cdot\|_\star=\|\cdot\|_{m,\pi}$, for some real number $m\in (2,+\infty]$. Fix $\I_0\subseteq \I$, and $\{\cB_i,\;i\in\I_0\}$ a family of nonempty  measurable subsets of $\Xset$, and set $\bar\cB \eqdef \cup_{i\in\I_0}\{i\}\times \cB_i$. Fix $\kappa>0$, and let a graph on $\I_0$ be such that
\[\int_{\cB_{i}\cap\cB_j}\min\left(\frac{\pi_{i}(x)}{\pi_i(\cB_i)},\frac{\pi_j(x)}{\pi_j(\cB_j)}\right) \rmd x\geq \kappa,\]
whenever there is an edge between  $i,j\in\I_0$. Let $\textsf{D}(\I_0)$ denote the diameter of the graph. Given $\zeta\in (0,1/2)$, if $\bar\pi(\bar\cB) \geq 1-\left(\frac{\zeta}{10}\right)^{1 + \frac{2}{m-2}}$, then
\[\textsf{SpecGap}_\zeta(K) \geq \frac{\kappa}{2\textsf{D}(\I_0)} \min_{i\in\I_0}\left\{\pi_i(\cB_i)^2\right\}\min_{i\in\I_0}\left\{\pi(i)\textsf{SpecGap}_{\cB_i}(K_i)\right\}.\]
\end{theorem}
\begin{proof}
See Section \ref{sec:proof:thm:mixing}.
\end{proof}

Note the similarity with (\ref{bound:madras}).  However Theorem \ref{thm:mixing} allows us to restrict the analysis of the chain to the set  $\bar\cB$.  Theorem \ref{thm:mixing} is basically a mixture analog of Lemma \ref{lem:useful}. 
In the important special case where $K_i(x,\rmd y) = \pi_i(\rmd y)$, and one chooses $\cB_i=\Xset$, Theorem \ref{thm:mixing} shows that \[\textsf{SpecGap}_\zeta(K) \geq \frac{\kappa}{2\textsf{D}(\I_0)}\min_{i\in\I_0}\left\{\pi(i)\right\},\;\mbox{ whereas }\;\; \textsf{SpecGap}(K) \geq \frac{\kappa}{2\textsf{D}(\I)}\min_{i\in\I}\left\{\pi(i)\right\}. \]
As we show with the next example these two lower bounds can have very different dependence on the dimension of $\Xset$.

\section{Analysis of a Gibbs sampler}\label{sec:ex:linregr}
We consider the Bayesian treatment of a linear regression problem with response variable $z\in\rset^n$,  and covariate matrix $X\in\rset^{n\times p}$. The regression parameter is denoted $\theta\in\rset^p$. In settings where the number of regressors $p$ is very large, and one is interested in selecting the most significant regressors and the corresponding coefficients, it is common practice to introduce an additional variable selection parameter $\delta\in\Delta\eqdef\{0,1\}^p$, and to use a spike-and-slab prior distribution on  $\theta$.
More precisely, given  $\textsf{q}\in (0,1)$ we assume that the prior distribution of $\delta$ is given by
\[\omega_\delta  =  \textsf{q}^{\|\delta\|_0}(1-\textsf{q})^{p-\|\delta\|_0},\;\;\delta\in\Delta,\]
and given $\rho,\gamma\in (0,+\infty)$, we assume that the components of $\theta$ are conditionally independent given $\delta$, and we assume that $\theta_j\vert \delta$ has density $\textbf{N}(0,\frac{1}{\rho})$ if $\delta_j=1$, and density $\textbf{N}(0,\gamma)$  otherwise, where $\textbf{N}(\mu,v^2)$ denotes the univariate Gaussian distribution with mean $\mu$ and variance $v^2$. The resulting posterior distribution on $\Delta\times\rset^p$ is
\begin{equation}\label{Pi:check}
\Pi(\delta,\rmd \theta\vert z)\propto \omega_\delta \frac{e^{ -\frac{1}{2}\theta' D_{(\delta)}^{-1} \theta}}{\sqrt{\det\left(2\pi D_{(\delta)}\right)}}  e^{-\frac{1}{2\sigma^2}\|z-X\theta\|_2^2} \rmd \theta,
\end{equation}
where $D_{(\delta)}\in\rset^{p\times p}$ is a diagonal matrix with $j$-th diagonal element equal to $1/\rho$ if $\delta_j=1$, and $\gamma$ if $\delta_j=0$. The regression error $\sigma$ is assumed known. This model is very popular in the application (\cite{george:mcculloch97,ishwaran:rao:2005,narisetti:he:14}), mainly because it is straightforward to sample from (\ref{Pi:check}). Indeed,  the posterior conditional distribution $\Pi(\delta\vert \theta,z)$ is a product of independent Bernoulli distributions, with closed form probabilities:
\begin{multline}\label{eq:pj}
\Pi(\delta\vert \theta,z) =  \prod_{j=1}^p\left[\textsf{q}_{j}\right]^{\delta_j}\left[1- \textsf{q}_{j}\right]^{1-\delta_j},\;\\
\;\;\mbox{ where }\; \textsf{q}_{j}\eqdef \frac{1}{1+\frac{1-\textsf{q}}{\textsf{q}}\sqrt{\frac{1}{\gamma\rho}}e^{\frac{1}{2}\left(\rho-\frac{1}{\gamma}\right)\theta_j^2}},\;j=1,\ldots,p.\end{multline}
 Given $\delta$,  the conditional distribution of $\theta$ given $\delta$ is $\textbf{N}_p(m_\delta,\sigma^2\Sigma_\delta)$, with $m_\delta$ and $\Sigma_\delta$ given by
\begin{equation}\label{mu:Sigma}
m_\delta \eqdef \Sigma_\delta X'z \;\;\mbox{ and }\;\; \Sigma_\delta \eqdef \left(X'X + \sigma^2D_{(\delta)}^{-1}\right)^{-1}.
\end{equation}
 Put together these two conditional distributions yield a simple Gibbs sampling algorithm for (\ref{Pi:check}).  We consider the following version that is modified so that the resulting Markov chain is lazy as  required by our theory. 
 
\begin{algorithm}\label{algo1}
For some initial distribution $\nu_0$ on $\rset^p$, draw $u_0\sim \nu_0$. Given $u_0,\ldots,u_k$ for some $k\geq 0$, draw independently $I_{k+1}\sim\textsf{Ber}(0.5)$.
\begin{enumerate}
\item If $I_{k+1}=0$, set $u_{k+1} = u_k$.
\item If $I_{k+1}=1$,
\begin{enumerate}
\item Draw $\delta\sim\Pi(\cdot\vert u_k,z)$ as given in (\ref{eq:pj}), and 
\item draw  $u_{k+1}\sim \textbf{N}_p(m_\delta,\sigma^2\Sigma_\delta)$ as given in (\ref{mu:Sigma}).
\end{enumerate}
\end{enumerate}
\vspace{-0.3cm}
\begin{flushright}
$\square$
\end{flushright}
\end{algorithm}

We analyze the mixing time of the marginal chain $\{u_k,\;k\geq 0\}$ from Algorithm \ref{algo1}. As easily seen, $\{u_k,\;k\geq 0\}$ is a Markov chain with invariant distribution
\begin{equation}\label{def:Pi:lm}
\Pi(\rmd\theta\vert z) \propto \sum_{\delta\in\Delta}\omega_\delta  \frac{e^{ -\frac{1}{2}\theta' D_{(\delta)}^{-1} \theta}}{\sqrt{\det\left(2\pi D_{(\delta)}\right)}}  e^{-\frac{1}{2\sigma^2}\|z-X\theta\|_2^2} \rmd \theta,
\end{equation}
which is of the form (\ref{pi:mix}), and with transition kernel 
\begin{equation}\label{def:cP}
K(u,\rmd \theta) \eqdef    \sum_{\omega\in\Delta}\Pi(\omega\vert u,z) \left[\frac{1}{2}\delta_u(\rmd \theta) + \frac{1}{2}\Pi(\rmd \theta\vert \omega,z)\right],
\end{equation}
which is of the form (\ref{K:mix}). In order to bring in the discussion the idea of posterior contraction toward  a true value of the parameter, we need to assume a model for the data.  

\begin{assumption}\label{H:lin:mod}
\begin{enumerate}
\item The data $z\in\rset^n$ is the realization of a random variable $Z \eqdef (Z_1,\ldots,Z_n) \sim \textbf{N}(X\theta_\star,\sigma^2 I_n)$, for some unknown parameter $\theta_\star\in\rset^p$, and a known absolute constant $\sigma^2>0$. 
\item The matrix $X$ is non-random and normalized such that
\begin{equation}\label{eq:norm:X}
\|X_j\|_2^2 = n, \;\;\;j=1,\ldots,p,
\end{equation}
where $X_j\in\rset^n$ denotes the $j$-th column of $X$. 
\item The prior parameter $\textsf{q}$ is chosen such that 
\begin{equation}
\label{choice:q}
\frac{\textsf{q}}{1-\textsf{q}} = \frac{1}{p^{u+1}},\end{equation}
for some absolute constant $u>0$.
\item The prior parameters $\rho$ and $\gamma$ are taken such that
\begin{equation}\label{rho:lm}
0 <\gamma < \frac{1}{2\rho}.
\end{equation}
\end{enumerate}
We will write $\PP_\star$ (resp. $\PE_\star$) to denote the probability distribution (resp. expectation operator) of the random variable $Z$ assumed in H\ref{H:lin:mod}. 
\end{assumption}
\medskip

\begin{remark} Overall these are very basic assumptions. We assume in H\ref{H:lin:mod}-(1) that the statistical model is well specified, and there is a true value of the parameter denoted $\theta_\star$. The assumption that the regression errors are Gaussian is imposed mostly for simplicity, and can be replaced by a sub-Gaussian assumption,  with minimal change to what follows.
The prior assumption in H\ref{H:lin:mod}-(3) is fairly standard, and follows \cite{castillo:etal:14,narisetti:he:14,atchade:15b}. 
H\ref{H:lin:mod}-(4)  simply says that the variance of the slab prior density should be sufficiently larger than the variance of the spike prior density.
\end{remark}

To proceed we introduce some notations. For $\theta,\theta'\in\rset^p$, we write $\theta\cdot\theta'\in\rset^p$ to represent the component-wise product of $\theta$ and $\theta'$.  For $\delta\in\Delta$, and $\theta\in\rset^p$, we  write $\theta_\delta$ as a short for $\theta\cdot\delta$, and we define $\delta^c\eqdef 1-\delta$, that is $\delta_j^c=1-\delta_j$, $1\leq j\leq p$. For a matrix $A\in\rset^{q\times p}$, $A_\delta$ (resp. $A_{\delta^c}$) denotes the matrix of $\rset^{q\times \|\delta\|_0}$ (resp. $\rset^{q\times(p-\|\delta\|_0)}$) obtained by keeping only the columns of $A$ for which $\delta_j=1$ (resp. $\delta_j=0$). For two elements $\delta,\delta'$ of $\Delta$, we write $\delta\supseteq\delta'$ to mean that $\delta_j=1$ whenever $\delta'_j=1$. The support of a vector $u\in\rset^p$ is the vector $\textsf{supp}(u)\in\Delta$ such that $\textsf{supp}(u)_j=1$ if and only if $|u_j|>0$.  

An important role is played in the analysis by the matrices
\[L_\delta\eqdef I_n + \frac{1}{\sigma^2}X D_{(\delta)}X'= I_n +\frac{1}{\sigma^2\rho}\sum_{j:\;\delta_j=1} X_jX_j' + \frac{\gamma}{\sigma^2} \sum_{j:\;\delta_j=0} X_jX_j',\;\;\;\;\;\;\delta\in\Delta,\]
and the  coherence of the  matrix $X$ defined for an integer $s\geq 1$ as 
\[\Cset(s) \eqdef \max_{\delta\in\Delta:\;\|\delta\|_0\leq  s }\;\;\max_{j\neq \ell}\;\;  \left|X_j' L_{\delta}^{-1}X_\ell\right|.\]

We will also need the following important assumption.

\begin{assumption}\label{H:tech}
There exist $\varrho>0$ and an integer  $s_0\in\{1,\ldots,p-1\}$, such that
\[\min_{\delta:\;\|\delta\|_0\leq s_0}\;\inf\left\{ \frac{u' \left(X_{\delta^c}'L_\delta^{-1} X_{\delta^c}\right) u}{n\|u\|_2^2},\;u\in\rset^{p-\|\delta\|_0},\; 0<\|u\|_0 \leq s_0\right\}\geq \varrho.\]
\end{assumption}
\medskip

\begin{remark}
For $\gamma$ small enough and $\rho \lesssim \sqrt{n\log(p)}/s$, we show in the appendix that the matrix $L_\delta^{-1}$ can be loosely interpreted  as the projector on the orthogonal of the space spanned by the columns of $X_\delta$. Therefore,  H\ref{H:tech} rules out settings  where a small number of columns of $X$ have the same linear span as all the columns of $X$. Indeed signal recovery becomes nearly impossible in such settings. We show in Lemma \ref{lem:H:tech} in the appendix that if $X$ is a random matrix with i.i.d. standard normal entries (Gaussian ensemble) and $\gamma$ is taken small enough, then H\ref{H:tech} holds with high probability, and 
\[\Cset(s) \leq c_0 \sqrt{n\log(p)},\]
for some universal constant $c_0$, provided that $n \gtrsim  s^2\log(p)$.
\vspace{-0.6cm}
\begin{flushright}
$\square$
\end{flushright}
\end{remark}

We need few more quantities in order to state the theorem. We define
\begin{equation}\label{def:eps:lm}
\epsilon \eqdef \sigma\sqrt{\frac{\log(p)}{n}},\end{equation}
that we view as the signal detectability threshold.  Let $\tilde \delta_\star$ be the element of $\Delta$ that indicates which components of $\theta_\star$ are greater than $\epsilon$ in absolute value (detectable components): $\tilde\delta_{\star,j} =1$ if and only if $|\theta_{\star,j}|>\epsilon$. Components of $\theta_\star$ that are below $\epsilon$ are too small to be detected. This  implies that the element of $\Delta$ toward which we can expect $\Pi(\cdot\vert z)$ to contract is $\tilde\delta_\star$ (here $\Pi(\cdot\vert z)$ refers to the $\delta$-marginal of the joint posterior). We formalize this contraction as follows. Given $k\geq 0$, we define 
\[\mathcal{D}_{k} \eqdef\left\{ \delta\in\Delta:\; \delta\supseteq \tilde\delta_\star,\; \|\delta\|_0 \leq \|\tilde\delta_\star\|_0 + k\right\},\]
which collects models that contain the true model (that is $\tilde\delta_\star$) and have at most $k$ false-positives, and we say that posterior contraction holds if 
\begin{equation}\label{post:contrac}
\Pi(\D_k\vert z)\geq 1-\frac{4}{p^{\frac{u}{2}(k+1)}}.
\end{equation}
We will not directly establish (\ref{post:contrac}). However several existing work suggest that this description of the posterior contraction of $\Pi(\cdot\vert z)$ holds. For instance under similar assumptions as above, \cite{narisetti:he:14} show that $\Pi(\D_0\vert Z)\geq 1-\frac{a_1}{p^{a_2}}$ with high-probability for positive constants $a_1,a_2$.  And when $\tilde\delta_\star = \delta_\star$, \cite{AB:18} shows that (\ref{post:contrac}) holds for a slightly modified version of the posterior distribution (\ref{Pi:check}). We introduce the event
\begin{multline*}
\e_k \eqdef\left\{z\in\rset^n:\;\Pi(\tilde\delta_\star\vert z)\geq 1/2,\;\;\Pi(\D_k\vert z)\geq 1-\frac{4}{p^{\frac{u}{2}(k+1)}},\;\right.\\
\left.\mbox{ and }\;\; \max_{\delta\supseteq\tilde\delta_\star:\;\|\delta\|_0\leq \tilde s_\star+ k}\;\sup_{1\leq j\leq p} \frac{1}{\sigma}\left|\pscal{L_{\delta}^{-1}X_j}{z-X\theta_\star}\right|\leq 2\sqrt{(k+1) n\log(p)}\right\}.\end{multline*}
Note that Gaussian tail bounds easily implies that under  H\ref{H:lin:mod}, the last part of $\e_k$ holds true with high probability. Hence the key condition in $\e_k$ is the posterior contraction assertion that $\Pi(\D_k\vert z) \geq 1-4p^{-u(k+1)/2}$. Here is our main result in this section.

\begin{theorem}\label{thm:mixing:lm}
Suppose  that H\ref{H:lin:mod} and H\ref{H:tech} hold and  Algorithm \ref{algo1} is initialized from $\nu_0 =\Pi(\cdot\vert \delta^{(\textsf{i})}, z)$, for $\delta^{(\mathsf{i})}$ that satisfies $\delta^{(\mathsf{i})}\supseteq\tilde\delta_\star$, with a number of false-positives  $\textsf{FP} \eqdef \|\delta^{(\mathsf{i})}\|_0 - \|\tilde\delta_\star\|_0$. Fix $\zeta_0\in (0,1)$. If  the dataset $z$ belongs to $\e_k$ for some $k\in \{0,\ldots,s_0\}$ that satisfies 
\begin{equation}\label{cond:v}
k +1\geq 4\left(1 +\frac{1}{u}\right)\textsf{FP} +\frac{2 \textsf{FP}}{u} \times \frac{ \log\left(1 + \frac{n \textsf{FP}}{\sigma^2\rho}\right)}{\log(p)}  + \frac{2}{u}\times \frac{\log\left(\frac{320}{\zeta_0^2}\right)}{\log(p)},
\end{equation}
then  the following holds true. There exists a constant  $A$ that does not depend on $p$ nor $\zeta_0$ such that for 
\begin{equation}\label{cond:it:2}
 N\geq \frac{A}{(\gamma\rho)} \log\left(\frac{1}{\zeta_0}\right)p^{\frac{1}{2\varrho}\left(s_\star+ 2\sqrt{1+k} + \frac{\|\tilde\theta_\star\|_1\Cset(\tilde s_\star +k)}{\sigma\sqrt{n\log(p)}}\right)^2} p^{k(u+1)} \left(1+ \frac{n k}{\sigma^2\rho}\right)^{\frac{k}{2}},\end{equation}
we have
\[\|\nu_0 K^N -\Pi(\cdot\vert z)\|_\tv \leq \zeta_0.\]
\end{theorem}
\begin{proof}
See Section \ref{sec:proof:thm:mixing:lm}.
\end{proof}

\subsection{Discussion}
Since we impose $k\leq s_0$, the condition (\ref{cond:v}) basically says  that the number of false-positives of $\delta^{(\textsf{i})}$ cannot be too large. Hence the main conclusion of Theorem \ref{thm:mixing:lm} is that  Algorithm \ref{algo1} has a polynomial mixing time  if posterior contraction holds ($z\in\e_k$), and the number of initial false-positives $\textsf{FP}$ is not too large --  the idea of warm-start. In contrast, the mixing time predicted by the standard spectral gap scales with $p$ as $O(p^p)$. This follows simply by plugging the lower bound (\ref{mino:pi}) in (\ref{bound:madras}). 

One of the first paper that analyzes the mixing times of MCMC algorithm in high-dimensional linear regression models  and highlights fast/slow mixing behaviors is \cite{yang:etal:15}. These authors takes a worst-case scenario approach\footnote{they look at the worst mixing time achievable by changing the initial distribution}, and show that in general their Gibbs sampler has a mixing time that is exponential in $p$ unless the state space is restricted to  only models $\delta$ for which $\|\delta\|_0\leq s_0$ for some threshold $s_0$. We note however that correctly choosing such threshold  $s_0$ in practice may be complicated. In contrast Theorem \ref{thm:mixing:lm} shows that without restricting the state space one can still achieve polynomial mixing time by warm-starting the algorithm. The idea of warm-start is a well-known  strategy to accelerate mixing times in MCMC computation (see e.g. \cite{lovasz:vempala07}).

 We note from (\ref{cond:v}) that the power $k$ that appears in  (\ref{cond:it:2}) grows with $\textsf{FP}$. This suggests that the mixing time of the algorithm can rapidly deteriorate as  $\textsf{FP}$ grows. It is unclear whether the precise dependence on $p$ thus expressed in (\ref{cond:it:2}) is tight.  In any case, we  did observe in the simulations a sharp increase in the mixing time of the algorithm as  \textsf{FP} increases, which seems consistent with (\ref{cond:it:2}).

With respect to the initialization, the natural question  is how the  mixing time behaves if $\delta^{(\mathsf{i})}$ admits false-negatives.  Our method is  not adapted to provide an answer to this question.  Nonetheless to gain some intuition, we perform some numerical simulations which seem to suggest that the polynomial mixing time obtained in Theorem \ref{thm:mixing:lm} no longer hold if $\delta^{(\mathsf{i})}$ has false-negatives.  

The bound in (\ref{cond:it:2}) highlights the effect of the coherence of the matrix $X$.  In general $\Cset(s)$ grows with $p$ as $\sqrt{\log(p)}$, which cancels with the same term in the denominator. However if there are strong correlations among some of the columns of $X$, then $\Cset(s)$ typically  grows with $n$ faster than $\sqrt{n}$, which can significantly impacts the mixing time. For instance if $n\gtrsim s\log(p)$ as assumed above, and $\Cset(s) \approx n$, then the resulting mixing time grows with $p$  faster than exponential.

Theorem \ref{thm:mixing:lm}  has also some obvious implications on how to initialize the chain. It suggests that the  initialization strategy sometimes used in practice where $\delta^{(\mathsf{i})}$ is taken as the zero vector is sub-optimal, and might result in Markov chain with exponential mixing times. Instead,  our result suggests a warm-start initialization where $\delta^{(\mathsf{i})}$ is taken for instance as the support of the lasso estimate -- or some other similarly-behaved frequentist estimate.

\subsection{Numerical illustrations}\label{sec:sim}
We illustrate some of the conclusions with the following simulation study. We consider a linear regression model with Gaussian noise $\textbf{N}(0,\sigma^2)$, where $\sigma^2$ is set to $1$. We experiment with sample size $n=p/10$, and dimension $p\in \{500, 1000, 2000, 3000,4000\}$.  We take $X\in\rset^{n\times p}$ as a random matrix with i.i.d. standard Gaussian entries. We fix the number of non-zero coefficients to $s_\star=10$, and $\delta_\star$ is given by
\[\delta_\star = (\underbrace{1,\ldots,1}_{10},\;\underbrace{0,\ldots,0}_{p-10}).\]
The non-zero coefficients of $\theta_\star$ are uniformly drawn from $(-a-1,-a)\cup (a,a+1)$, where 
\[a = 4\sqrt{\frac{\log(p)}{n}}.\] 

We use the following prior parameters values: 
\[ u = 1,\; \; \rho = \frac{1}{\sqrt{n}},\;\; \gamma = \frac{0.1\sigma^2}{\lambda_{\textsf{max}}(X'X)}.\]
We use an initial distribution $\nu_0=\Pi(\cdot\vert \delta^{(\mathsf{i})},z)$, where we vary the number of false-positives of $\delta^{(\textsf{i})}$. To monitor the mixing, we compute the sensitivity and the precision at iteration $k$ as
\begin{multline*}
\textsf{SEN}_k = \frac{1}{s_\star}\sum_{j=1}^p \textbf{1}_{\{|\delta_{k,j}|>0\}}\textbf{1}_{\{|\delta_{\star,j}|>0\}},\; \textsf{PREC}_k=\frac{\sum_{j=1}^p \textbf{1}_{\{|\delta_{k,j}|>0\}}\textbf{1}_{\{|\delta_{\star,j}|>0\}}}{\sum_{j=1}^p \textbf{1}_{\{|\delta_{k,j}|>0\}}}. \end{multline*}
We empirically measure the mixing time of the algorithm as the first time $k$ where both $\textsf{SEN}_k$ and $\textsf{PREC}_k$ reach $1$, truncated to $2\times 10^4$ -- that is we stop any run that has not mixed after $20000$ iterations. The average empirical mixing time thus obtained (based on on $50$ independent MCMC replications) are presented in Table \ref{table:1} and Figure \ref{fig:1}. These estimates are consistent with our results. They show only a modest increase in mixing time as $p$ increases, but a sharp increase in mixing time as the number of false-positives increases. We also explore the behavior of the sampler in the presence of false-negatives in the initialization. More specifically we consider the case where $\delta^{(\textsf{i})}$ has 2 false-negatives, but no false-positive. In this setting, and for all 50 replications, the sampler fails to recover all $10$ significant components within $20,000$ iterations. 
\medskip

\begin{table}[h]
\begin{center}
\small
\scalebox{.9}{\begin{tabular}{llllll}
\hline
 & $p=500$ &  $p=1000$ &   $p=2000$ &  $p=3000$ & $p=4000$ \\
\hline
$\textsf{FP}=1\%$ & 15.7 (21.1) & 71.6 (280.3) & 43.5 (45.6) &  42.8 (47.5) & 65.4 (100.6)\\
$\textsf{FP}=5\%$& 93.8 (247.8) & 93.5 (102.3) & 130.8 (164.9) & 186.9 (303.9) & 225.9 (239.5)\\
$\textsf{FP}=10\%$ &  $>$11325.6  & $>$7916.3 & $>$8955.0  & $>$10648.0 & $>$12113.4 \\
$\textsf{FP}=20\%$ &  $>$20000  & $>$20000 & $>$20000  & $>$20000 & $>$20000\\
$\textsf{FP}=0$, $\textsf{FN}=2$ &  $>$20000  & $>$20000 & $>$20000  & $>$20000 & $>$20000\\
\hline
\end{tabular}}
\caption{\small{Table showing the average empirical mixing time of the sampler. Based on 50 simulation replications. The numbers in parenthesis are standard errors. The notation $>a$ means that some (or all) of the replicated mixing times have been truncated to $20,000$.
}}\label{table:1}
\end{center}
\end{table}

\medskip

\begin{center}

\begin{figure}[h!]
\centering
\includegraphics[scale=0.5]{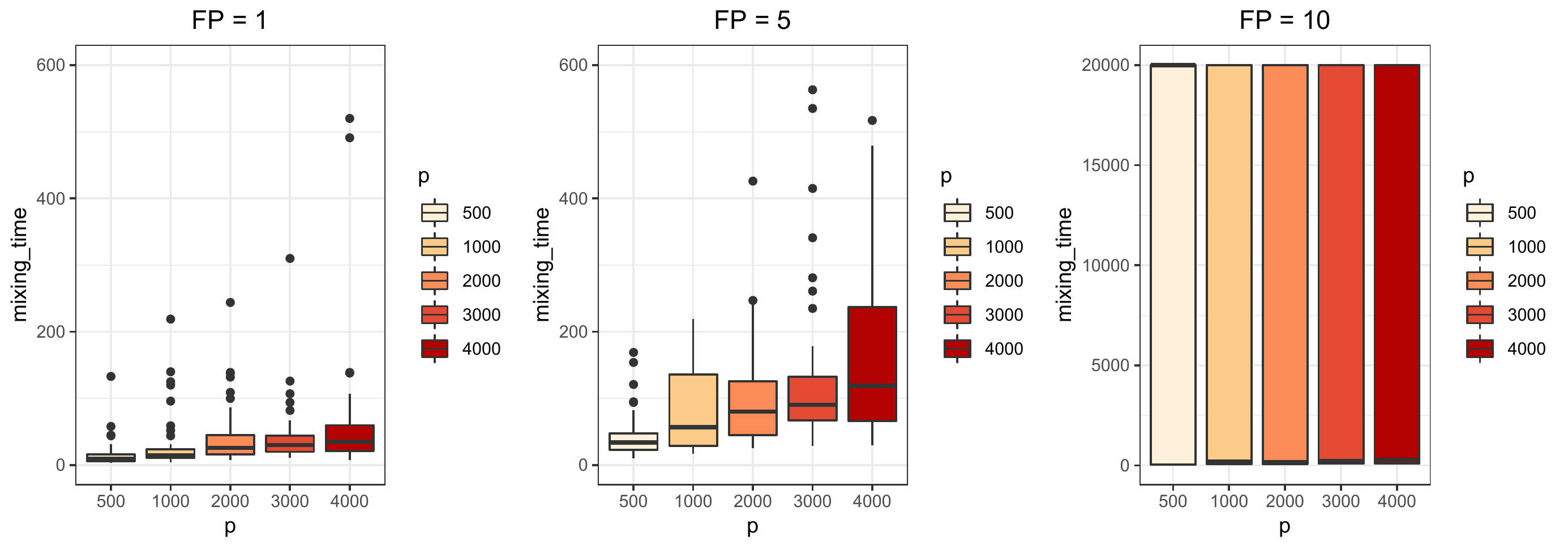}
\vspace{0.0cm}\caption{{\small{Boxplots of the average empirical mixing times. Based on 50 simulation replications. $\textsf{FP}=x$ means there are $p\times x/100$ false-positives.} }}
\label{fig:1}
\end{figure}

\end{center}

%

\section{Proofs}
\subsection{Proof Lemma \ref{lem:key}}\label{sec:proof:lem:key}
We first note that if a probability measure  $\nu$ is absolutely continuous with respect to $\pi$ with Radon-Nikodym  derivative $f_{\nu}$, then for any $A\in\B$,
\begin{eqnarray*}
\nu K(A) & = & \int\nu(\rmd x)K(x,A) =  \int \int f_\nu(x)\textbf{1}_A(y)\pi(\rmd x)K(x,\rmd y)\\
&=& \int \int \textbf{1}_A(x)f_\nu(y)\pi(\rmd x)K(x,\rmd y) =  \int_A \pi(\rmd x) \int K(x,\rmd y) f_\nu(y),
\end{eqnarray*}
where the third equality uses the reversibility of $K$. This calculation says that $\nu K$ is also absolutely continuous with respect to $\pi$ with Radon-Nikodym  derivative $x\mapsto K f_\nu(x)\eqdef \int K(x,\rmd y) f_\nu(y)$.  More generally $\frac{\rmd (\nu K^n)}{\rmd \pi}(\cdot) = K^nf_\nu(\cdot)$, and 
\begin{eqnarray}\label{keylem:eq1}
\|\nu K^n -\pi\|_\tv^2 & = & \left(\int \left|\frac{\rmd (\nu K^n)}{\rmd \pi}(x) - 1\right|\pi(\rmd x)\right)^2 \nonumber\\
& =& \left(\int \left|K^nf_\nu(x) - 1\right|\pi(\rmd x)\right)^2\nonumber \\
& \leq & \| K^nf_\nu-1\|_{2,\pi}^2 \nonumber\\
& = & \textsf{Var}( K^nf_\nu).\end{eqnarray}
Take $f\in L^2(\pi)$. Since $\pi(f) = \pi(Kf)$, we have
\begin{multline}\label{keylem:eq11}
\textsf{Var}(Kf) - \textsf{Var}(f) = \pscal{Kf}{Kf}_\pi - \pscal{f}{f}_\pi \\
= -\frac{1}{2}\int\int\left(f(y)-f(x)\right)^2\pi(\rmd x) K^2(x,\rmd y),\end{multline}
where the last equality exploits the reversibility of $K$. For any function $f\in L^2(\pi)$,
\begin{multline*}
\int\pi(\rmd x) \int K^2(x,\rmd y) (f(y)-f(x))^2 = \int\pi(\rmd x)\int_{\Xset} K(x,\rmd x_1) \int K(x_1,\rmd y)(f(y)-f(x))^2\\
= \int\pi(\rmd x)\int_{\{x\}} K(x,\rmd x_1) \int K(x_1,\rmd y)(f(y)-f(x))^2  \\
+ \int\pi(\rmd x)\int_{\Xset\setminus \{x\}} K(x,\rmd x_1) \int K(x_1,\rmd y)(f(y)-f(x))^2.
\end{multline*}
By the lazyness of the chain, the first term on the right hand side of the last display is bounded from below by
\[\frac{1}{2}\int\pi(\rmd x)\int K(x,\rmd y)(f(y)-f(x))^2,\]
whereas the second term  is bounded from below by
\begin{multline*}
\int\pi(\rmd x)\int_{\Xset\setminus \{x\}} K(x,\rmd x_1) \int_{\{x_1\}} K(x_1,\rmd y)(f(y)-f(x))^2\\
\geq \frac{1}{2} \int\pi(\rmd x)\int_{\Xset\setminus \{x\}} K(x,\rmd x_1)(f(x_1)-f(x))^2 = \frac{1}{2} \int\pi(\rmd x)\int K(x,\rmd x_1)(f(x_1)-f(x))^2.
\end{multline*}
Hence, for all $f\in L^2(\pi)$,
\[
\int\int\left(f(y)-f(x)\right)^2\pi(\rmd x) K^2(x,\rmd y) \geq  \int\int\left(f(y)-f(x)\right)^2\pi(\rmd x) K(x,\rmd y).\]
Using the last display together with (\ref{keylem:eq11}),  and the definition of $\e(f,f)$, we conclude that for all $f\in L^2(\pi)$,
\begin{equation}\label{keylem:eq2}
\textsf{Var}(Kf) \leq \textsf{Var}(f)  -\e(f,f).
\end{equation}
Fix $\zeta\in (0,1)$, and take $f\in L^2(\pi)$. If $\textsf{Var}(f)\leq \zeta\|f\|_\star^2$, then, by (\ref{keylem:eq2}),
\begin{multline*}
\textsf{Var}(Kf) \leq \textsf{Var}(f) \leq \zeta\|f\|_\star^2 \\
 = \left(1-\textsf{SpecGap}_\zeta(K)\right) \max\left(\textsf{Var}(f),\zeta\|f\|_\star^2\right) + \textsf{SpecGap}_\zeta(K)\zeta\|f\|_\star^2.
 \end{multline*}
 But if $\textsf{Var}(f) > \zeta\|f\|_\star^2>0$, then by (\ref{keylem:eq2}), 
\begin{eqnarray*}
\textsf{Var}(Kf) & = & \|f\|_\star^2\textsf{Var}\left(K\left(\frac{f}{\|f\|_\star}\right) \right) \\
 & \leq &  \|f\|_\star^2\left(\textsf{Var}\left(\frac{f}{\|f\|_\star}\right) - \e\left(\frac{f}{\|f\|_\star},\frac{f}{\|f\|_\star}\right)\right)\\
& \leq & \textsf{Var}(f) - \|f\|_\star^2\textsf{SpecGap}_\zeta(K)\left(\textsf{Var}\left(\frac{f}{\|f\|_\star}\right)  - \frac{\zeta}{2}\right),\\
& = & \textsf{Var}(f)\left(1 - \textsf{SpecGap}_\zeta(K)\right) + \frac{\zeta}{2}\|f\|_\star^2\textsf{SpecGap}_\zeta(K).
\end{eqnarray*}
Clearly the last display (which is derived assuming that $\|f\|_\star>0$)  continues to hold if $\|f\|_\star=0$. We conclude that for all $f\in L^2(\pi)$, 
\[\textsf{Var}(Kf) \leq \max\left(\textsf{Var}(f),\zeta\|f\|_\star^2\right) \left(1 - \textsf{SpecGap}_\zeta(K)\right) + \zeta\|f\|_\star^2\textsf{SpecGap}_\zeta(K).\]
Since $\|Kf\|_\star\leq \|f\|_\star$, it follows that for all $f\in L^2(\pi)$
\begin{multline}\label{keylem:eq3}
\max\left(\textsf{Var}(Kf),\zeta\|K f\|_\star^2\right)  \\
\leq \max\left(\textsf{Var}(f),\zeta\|f\|_\star^2\right) \left(1 - \textsf{SpecGap}_\zeta(K)\right) 
+ \zeta\|f\|_\star^2\textsf{SpecGap}_\zeta(K).\end{multline}
We can iterate the above inequality to deduce that for all $f\in L^2(\pi)$, such that $\|f\|_\star<\infty$, and for all $n\geq 1$,
\begin{multline*}
\max\left(\textsf{Var}(K^nf),\zeta\|K^n f\|_\star^2\right) \leq  \max\left(\textsf{Var}(f),\zeta\|f\|_\star^2\right)\left(1 - \textsf{SpecGap}_\zeta(K)\right)^n \\
+ \zeta\textsf{SpecGap}_\zeta(K)\sum_{j\geq 0} \left(1 - \textsf{SpecGap}_\zeta(K)\right)^j\|K^{n-j-1}f\|_\star^2\\
\leq \max\left(\textsf{Var}(f),\zeta\|f\|_\star^2\right)\left(1 - \textsf{SpecGap}_\zeta(K)\right)^n + \zeta\|f\|_\star^2.
\end{multline*}
Now, if $\pi_0= f_0 \pi$, the last display combined with (\ref{keylem:eq1}) implies that
\begin{multline*}
\|\pi_0 K^n -\pi\|_\tv^2 \leq \max\left(\textsf{Var}(K^nf_0),\zeta\|K^n f_0\|_\star^2\right) \\
\leq  \max\left(\textsf{Var}(f_0),\zeta\|f_0\|_\star^2\right)\left(1 - \textsf{SpecGap}_\zeta(K)\right)^n + \zeta\|f_0\|_\star^2.\end{multline*}
This ends the proof.
\vspace{-0.6cm}
\begin{flushright}
$\square$
\end{flushright}

\medskip
\subsection{Proof Lemma \ref{lem:useful}}\label{sec:proof:lem:useful}
Take $f:\;\Xset\to\rset$ such that $\textsf{Var}_\pi(f)>\zeta$, and $\|f\|_\star = \|f\|_{m,\pi}=1$. We have
\begin{multline*}
2\textsf{Var}_\pi(f) = \int_{\Xset_\zeta}\int_{\Xset_\zeta}(f(y)-f(x))^2\pi(\rmd x)\pi(\rmd y)  \\
+ 2\int_{\Xset_\zeta}\int_{\Xset\setminus\Xset_\zeta}(f(y)-f(x))^2\pi(\rmd x)\pi(\rmd y)
+ \int_{\Xset\setminus \Xset_\zeta}\int_{\Xset\setminus \Xset_\zeta}(f(y)-f(x))^2\pi(\rmd x)\pi(\rmd y).\end{multline*}
Using the convexity inequality $(a+b)^2\leq 2a^2 +2b^2$, and  Holder's inequality,
\begin{multline*}
\int_{\Xset_\zeta}\int_{\Xset\setminus\Xset_\zeta}(f(y)-f(x))^2 \pi(\rmd x)\pi(\rmd y)\\
 \leq 2 \pi(\Xset_\zeta)\int_{\Xset\setminus\Xset_\zeta} f(x)^2\pi(\rmd x) + 2\pi(\Xset\setminus\Xset_\zeta)\int_{\Xset_\zeta}f(x)^2\pi(\rmd x)\\
\leq 2 \pi(\Xset_\zeta) \pi(\Xset\setminus\Xset_\zeta)^{1-\frac{2}{m}}\|f\|_{m,\pi}^2 +  2\pi(\Xset\setminus\Xset_\zeta)\|f\|_{m,\pi}^2 \\
\leq 4\pi(\Xset\setminus\Xset_\zeta)^{1-\frac{2}{m}}. \end{multline*}
With similar calculation,
\[\int_{\Xset\setminus \Xset_\zeta}\int_{\Xset\setminus \Xset_\zeta}(f(y)-f(x))^2 \pi(\rmd x)\pi(\rmd y) \leq 4 \pi(\Xset\setminus\Xset_\zeta) \pi(\Xset\setminus\Xset_\zeta)^{1-\frac{2}{m}} \leq 2\pi(\Xset\setminus\Xset_\zeta)^{1-\frac{2}{m}}.\]
Using $\pi(\Xset_\zeta)\geq (\zeta/5)^{1+ 2/(m-2)}$, we get \[2(\textsf{Var}_\pi(f)-\frac{\zeta}{2}) \geq \int_{\Xset_\zeta}\int_{\Xset_\zeta}\pi(\rmd x)\pi(\rmd y)(f(y)-f(x))^2.\]
Hence
\[\frac{\e(f,f)}{\textsf{Var}_\pi(f)-\frac{\zeta}{2}} \geq \frac{\int_{\Xset_\zeta}\int_{\Xset_\zeta}\pi(\rmd x)K(x,\rmd y)(f(y)-f(x))^2}{\int_{\Xset_\zeta}\int_{\Xset_\zeta}\pi(\rmd x)\pi(\rmd y)(f(y)-f(x))^2}\geq \textsf{SpecGap}_{\Xset_\zeta}.\]
This ends the proof.
\vspace{-0.6cm}
\begin{flushright}
$\square$
\end{flushright}

\medskip

\subsection{Proof Theorem \ref{thm:mixing}}\label{sec:proof:thm:mixing}
The proof of the theorem is similar to the proof of Lemma \ref{lem:useful}.  But first, we need the following lemma.
\begin{lemma}\label{mixing:lem2}
Let $\nu(\rmd x) =f_\nu(x)\rmd x$, $\mu(\rmd x) = f_\mu(x) \rmd x$ be two probability measures on some measurable space with reference measure $\rmd x$, such that $\int \min(f_\mu(x),f_\nu(x))\rmd x>\epsilon$ for some $\epsilon>0$. Then for any measurable function $h$ such that $\int h^2(x)\nu(\rmd x)<\infty$ and $\int h^2(x)\mu(\rmd x)<\infty$, we have
\begin{multline*}
\int (h(y)-h(x))^2\mu(\rmd y)\nu(\rmd x) \\
\leq \frac{2-\epsilon}{2\epsilon} \left[\int (h(y)-h(x))^2\mu(\rmd y)\mu(\rmd x) + \int (h(y)-h(x))^2\nu(\rmd y)\nu(\rmd x)\right].\end{multline*}
\end{lemma}
\begin{proof}
This result is established as part of the proof of Theorem 1.2 of \cite{madras:randall:02} (see inequality (47)). 
\end{proof}
Choose $f\in L^2(\pi)$ such that $\|f\|_{m,\pi}=1$ . Given $i\in\I$, we set 
\[\e_i(f,f)\eqdef \frac{1}{2}\int_{\cB_i}\int_{\cB_i}\left(f(y)-f(x)\right)^2\pi_i(\rmd x)K_i(x,\rmd y).\]
By the definition of $\textsf{SpecGap}_{\cB_i}(K_i)$, we have
\begin{equation}\label{proof:thm:mixing:2:eq1}
\e_i(f,f) \geq \frac{1}{2} \textsf{SpecGap}_{\cB_i}(K_i) \int_{\cB_i}\int_{\cB_i}\left(f(y)-f(x)\right)^2\pi_i(\rmd x)\pi_i(\rmd y).\end{equation}

By Fubini's theorem, and using (\ref{proof:thm:mixing:2:eq1}), we have
\begin{eqnarray}\label{proof:thm:mixing:2:eq2}
2\e(f,f) &=& \int_{\Xset}\pi(\rmd x)\left\{\sum_{i\in\I} \pi_x(i) \int_{\Xset} K_i(x,\rmd y) \right\}\left(f(y)-f(x)\right)^2 \nonumber\\
& =& \sum_{i\in\I}\pi(i)\int_{\Xset}\int_{\Xset}\left(f(y)-f(x)\right)^2\pi_i(\rmd x)K_i(x,\rmd y) \nonumber\\
& \geq & 2\sum_{i\in\I_0}   \pi(i) \e_i(f,f) \nonumber\\
& \geq & \sum_{i\in\I_0} \pi(i) \textsf{SpecGap}_{\cB_i}(K_i) \int_{\cB_i}\int_{\cB_i}\left(f(y)-f(x)\right)^2\pi_i(\rmd x)\pi_i(\rmd y) \nonumber\\
& \geq & \min_{i\in\I_0}\left\{\pi(i)\pi_i(\cB_i)\textsf{SpecGap}_{\cB_i}(K_i)\right\}\nonumber\\
&& \times \sum_{i\in\I_0} \frac{1}{\pi_i(\cB_i)}\int_{\cB_i}\int_{\cB_i}\left(f(y)-f(x)\right)^2\pi_i(\rmd x)\pi_i(\rmd y)  .\end{eqnarray}

Using $\bar \cB = \cup_{i\in\I_0}\{i\}\times \cB_i$, and $\bar\cB^c\eqdef (\I\times \Xset)\setminus \bar\cB$, we write
\begin{eqnarray*}
2\textsf{Var}_\pi(f) & = & \int_{\I\times \Xset}\int_{\I\times \Xset}\left(f(y)-f(x)\right)^2\bar\pi(i,\rmd x)\bar \pi(j,\rmd y) \\
& \leq &  \int_{\bar\cB}\int_{\bar \cB}\left(f(y)-f(x)\right)^2\bar\pi(i,\rmd x)\bar \pi(j,\rmd y) + 10 \bar\pi(\bar \cB^c)^{1-\frac{2}{m}},
\end{eqnarray*}
by using similar calculations as in Lemma \ref{lem:useful}.
And since $\bar\cB$ is such that $5\bar\pi(\bar \cB^c)^{1-\frac{2}{m}}\leq \zeta$, we conclude that
\begin{eqnarray}\label{proof:thm:mixing:2:eq3}
2\left(\textsf{Var}_\pi(f) -\zeta\right) & \leq & \int_{\bar\cB}\int_{\bar \cB}\left(f(y)-f(x)\right)^2\bar\pi(i,\rmd x)\bar \pi(j,\rmd y),\nonumber\\
 & = & \sum_{i\in\I_0}\sum_{j\in \I_0} \pi(i)\pi(j) \int_{\cB_i}\int_{\cB_j}\left(f(y)-f(x)\right)^2\pi_i(\rmd x)\pi_j(\rmd y).\nonumber\\
 & = & \sum_{i\in\I_0}\sum_{j\in \I_0} \pi(i)\pi_i(\cB_i)\pi(j)\pi_j(\cB_j)\int_{\cB_i}\int_{\cB_j}\left(f(y)-f(x)\right)^2\frac{\pi_i(\rmd x)}{\pi_i(\cB_i)}\frac{\pi_j(\rmd y)}{\pi_j(\cB_j)}\nonumber\\
 \end{eqnarray}
For $i,j\in\I_0$, let us write $(i,j)$ to denote the path from $i$ to $j$, and given an edge $e$, let us write $e$ as $(e_1,e_2)$ where $e_1$ and $e_2$ denote the incident nodes of $e$. By the Cauchy-Schwarz inequality,
\begin{multline*}
\int_{\cB_i}\int_{\cB_j}\left(f(y)-f(x)\right)^2\frac{\pi_i(\rmd x)}{\pi_i(\cB_i)}\frac{\pi_j(\rmd y)}{\pi_j(\cB_j)} \leq \sum_{e\in(i,j)} \frac{1}{\min\left(\pi_{e_1}(\cB_{e_1}), \pi_{e_2}(\cB_{e_2})\right)} \\
\times \sum_{e\in (i,j)}\min\left(\pi_{e_1}(\cB_{e_1}), \pi_{e_2}(\cB_{e_2})\right) \int_{\cB_{e_1}}\int_{\cB_{e_2}} \left(f(y)-f(x)\right)^2\frac{\pi_{e_1}(\rmd x)}{\pi_{e_1}(\cB_{e_1})}\frac{\pi_{e_2}(\rmd y)}{\pi_{e_2}(\cB_{e_2})}.
\end{multline*}
By Lemma \ref{mixing:lem2}, integral on the right-hand side of the last display is upper bounded by
\begin{multline*}
\left(\frac{2-\kappa}{2\kappa}\right) \frac{1}{\pi_{e_1}(\cB_{e_1})^2}\int_{\cB_{e_1}}\int_{\cB_{e_1}}\left(f(y)-f(x)\right)^2\pi_{e_1}(\rmd x)\pi_{e_1}(\rmd y) \\
 + \left(\frac{2-\kappa}{2\kappa}\right) \frac{1}{\pi_{e_2}(\cB_{e_2})^2}\int_{\cB_{e_2}}\int_{\cB_{e_2}}\left(f(y)-f(x)\right)^2\pi_{e_2}(\rmd x)\pi_{e_2}(\rmd y).\end{multline*}
 Therefore the last inequality becomes
 \begin{multline*}
\int_{\cB_i}\int_{\cB_j}\left(f(y)-f(x)\right)^2\frac{\pi_i(\rmd x)}{\pi_i(\cB_i)}\frac{\pi_j(\rmd y)}{\pi_j(\cB_j)} \\
\leq \frac{\textsf{D}(\I_0)}{\min_{i\in\I_0}\pi_i(\cB_i)}\frac{2}{\kappa}\sum_{i\in\I_0} \frac{1}{\pi_i(\cB_i)}\int_{\cB_i}\int_{\cB_i}
\left(f(y)-f(x)\right)^2\pi_{i}(\rmd x)\pi_i(\rmd y).
\end{multline*}

This inequality together with  (\ref{proof:thm:mixing:2:eq3}) and  (\ref{proof:thm:mixing:2:eq2}) gives
\[\frac{\e(f,f)}{\textsf{Var}(f)-\zeta} \geq \frac{\kappa}{2\textsf{D}(\I_0)} \min_{i\in\I_0}\left\{\pi_i(\cB_i)^2\right\}\min_{i\in\I_0}\left\{\pi(i)\textsf{SpecGap}_{\cB_i}(K_i)\right\}.\]
This concludes the proof.
\vspace{-0.6cm}
\begin{flushright}
$\square$
\end{flushright}

\medskip

\subsection{Some preliminary remarks on the proof of Theorem \ref{thm:mixing:lm}}
We collect here some basic calculations on $\Pi(\cdot\vert z)$ that we rely on  repeatedly in the proofs. For any subset $B$ of $\Delta$, and $\delta_0\in\Delta$, we have
\begin{eqnarray}\label{eq1:proof:lm:step0}
\Pi(B\vert z) & = &   \Pi(\delta_0\vert z) \sum_{\delta\in B}\frac{\Pi(\delta\vert z)}{\Pi(\delta_0\vert z)}\nonumber\\
 & = &  \Pi(\delta_0\vert z) \sum_{\delta\in B}\frac{\omega_\delta}{\omega_{\delta_0}} \left(\gamma\rho\right)^{\frac{\|\delta\|_0-\|\delta_0\|_0}{2}} \frac{\int_{\rset^p}e^{-\frac{1}{2\sigma^2}\|z-Xu\|_2^2-\frac{1}{2}u'D_{(\delta)}^{-1} u}\rmd u}{\int_{\rset^p}e^{-\frac{1}{2\sigma^2}\|z-Xu\|_2^2-\frac{1}{2}u'D_{(\delta_0)}^{-1} u}\rmd u}\nonumber\\
& =& \Pi(\delta_0\vert z)\sum_{\delta\in B}\frac{\omega_\delta}{\omega_{\delta_0}} \left(\gamma\rho\right)^{\frac{\|\delta\|_0-\|\delta_0\|_0}{2}}\frac{\sqrt{\det\left(\sigma^2D_{(\delta_0)}^{-1}+X'X\right)}}{\sqrt{\det\left(\sigma^2D_{(\delta)}^{-1}+X'X\right)}}\nonumber\\
&& \times \frac{e^{\frac{1}{2\sigma^2}z'X\left(\sigma^2D_{(\delta)}^{-1}+X'X\right)^{-1}X'z}}{e^{\frac{1}{2\sigma^2}z'X\left(\sigma^2D_{(\delta_0)}^{-1}+X'X\right)^{-1}X'z}}. \end{eqnarray}
By the determinant lemma ($\det(A+UV') = \det(A)\det(I_m+V'A^{-1}U)$ valid for any invertible matrix $A\in\rset^{n\times n}$, and $U,V\in\rset^{n\times m}$) we have
\[
\left(\gamma\rho\right)^{\frac{\|\delta\|_0-\|\delta_0\|_0}{2}}\frac{\sqrt{\det\left(\sigma^2D_{(\delta_0)}^{-1}+X'X\right)}}{\sqrt{\det\left(\sigma^2D_{(\delta)}^{-1}+X'X\right)}}  = \sqrt{\frac{\det\left(I_n + \frac{1}{\sigma^2}X D_{(\delta_0)}X'\right)}{\det\left(I_n + \frac{1}{\sigma^2}X D_{(\delta)}X'\right)}}.\]
By the Woodbury identity (\cite{horn:johnson}~Section 0.7.4) which states that for any set of matrices $U,V,A,C$ with matching dimensions, $(A+UCV)^{-1} = A^{-1} -A^{-1}U(C^{-1} + VA^{-1}U)^{-1}VA^{-1}$, we have
\begin{multline*}
X\left(\sigma^2D_{(\delta)}^{-1}+X'X\right)^{-1}X' = \frac{1}{\sigma^2}XD_{(\delta)}X' - \frac{1}{\sigma^4}XD_{(\delta)}X'\left(I_n + \frac{1}{\sigma^2}XD_{(\delta)}X'\right)^{-1}XD_{(\delta)}X' \\
= I_n -\left(I_n +\frac{1}{\sigma^2}X D_{(\delta)}  X'\right)^{-1}.\end{multline*}
Hence,
\[
\frac{e^{\frac{1}{2\sigma^2}z'X\left(\sigma^2D_{(\delta)}^{-1}+X'X\right)^{-1}X'z}}{e^{\frac{1}{2\sigma^2}z'X\left(\sigma^2D_{(\delta_0)}^{-1}+X'X\right)^{-1}X'z}}  =  \frac{e^{\frac{1}{2\sigma^2}z'\left(I_n +\frac{1}{\sigma^2} X D_{(\delta_0)} X'\right)^{-1}z}}{e^{\frac{1}{2\sigma^2}z'\left(I_n +\frac{1}{\sigma^2}X D_{(\delta)} X'\right)^{-1}z}}.\]
It follows from the above and (\ref{eq1:proof:lm:step0}) that for all $\delta_0\in\Delta$, and $B\subseteq \Delta$,
\begin{equation}\label{eq2:proof:lm:step0}
\Pi(B\vert z) =\Pi(\delta_0\vert z) \sum_{\delta\in B} \frac{\omega_\delta}{\omega_{\delta_0}} \sqrt{\frac{\det\left(L_{\delta_0}\right)}{\det\left(L_{\delta}\right)}} \frac{e^{\frac{1}{2\sigma^4}z'L_{\delta_0}^{-1}z}}{e^{\frac{1}{2\sigma^4}z'L_{\delta}^{-1}z}},
\end{equation}
where, for $\delta\in\Delta$,  we recall the definition $L_\delta \eqdef I_n + \frac{1}{\sigma^2} X D_{(\delta)} X'$. We will use the following to deal with the terms involved in (\ref{eq2:proof:lm:step0}).  Suppose that we have $\vartheta,\delta\in\Delta$ such that $\vartheta\supseteq\delta$. Setting $\tau \eqdef \frac{1}{\sigma^2}\left(\frac{1}{\rho}- \gamma\right)$, it is easily seen that
\begin{equation}\label{eq1:control:term1}
L_{\vartheta} = L_{\delta} + \tau \sum_{j:\;\delta_{j}=0,\vartheta_j=1}\; X_jX_j'.\end{equation}
Therefore by the determinant lemma,
\begin{equation}\label{eq2:control:term1}
\frac{\det(L_\vartheta)}{\det(L_\delta)} =  \det\left(I_{\|\vartheta-\delta\|_0} + \tau X_{(\vartheta-\delta)}'L_\delta^{-1}X_{(\vartheta-\delta)}\right).\end{equation}
And by the Woodbury identity,
\begin{equation}\label{eq21:control:term1}
L_\vartheta^{-1} = L_\delta^{-1} -\tau L_\delta^{-1}X_{(\vartheta-\delta)}\left(I_{\|\vartheta-\delta\|_0} + \tau X_{(\vartheta-\delta)}'L_\delta^{-1}X_{(\vartheta-\delta)}\right)^{-1}X_{(\vartheta-\delta)}'L_{\delta}^{-1}.\end{equation}

\medskip

\subsection{Proof Theorem \ref{thm:mixing:lm}}\label{sec:proof:thm:mixing:lm}
Throughout, we fix $\zeta_0\in (0,1)$, and $z\in\e_k$ for some $k$ that satisfies (\ref{cond:v}).  We recall that the initial distribution is taken as $\nu_0 =\Pi(\cdot\vert\delta^{(\textsf{i})},z)$, for some initial choice $\delta^{(\textsf{i})}$.  Let 
 \[f_0(\theta)  \eqdef \frac{\nu_{0}(\theta)}{\Pi(\theta\vert z)},\;\theta\in\rset^p,\] 
be the density of $\nu_{0}$ with respect to $\Pi(\cdot\vert z)$. Since  $\Pi(\theta\vert z) \geq \Pi(\delta^{(\textsf{i})}\vert z) \Pi(\theta\vert \delta^{(\textsf{i})}, z)$,  we have
\[
f_0(\theta) = \frac{\Pi(\theta\vert \delta^{(\textsf{i})}, z)}{ \Pi(\theta\vert z)} \leq \frac{1}{\Pi(\delta^{(\textsf{i})}\vert z)}.\]
Using $\Pi(\tilde\delta_\star\vert z)\geq 1/2$ we can write,
\[
\frac{1}{\Pi(\delta^{(\textsf{i})}\vert z)} \leq  \frac{2\Pi(\tilde\delta_\star\vert z)}{\Pi(\delta^{(\textsf{i})}\vert z)}.\]
Using (\ref{eq2:proof:lm:step0}) with $B=\{\tilde\delta_\star\}$ and $\delta_0=\delta^{(\textsf{i})}$, and using  (\ref{eq2:control:term1}) and (\ref{eq21:control:term1}), we deduce that 
\begin{eqnarray*}
\|f_0\|_{\pi,\infty} \leq \frac{1}{\Pi(\delta^{(\textsf{i})}\vert z)} & \leq & 2 p^{(u+1)\textsf{FP}}\sqrt{\det\left(I_{\textsf{FP}} + \tau X_{(\delta^{(\textsf{i})}-\tilde\delta_\star)}' L_{\tilde\delta_\star}^{-1} X_{(\delta^{(\textsf{i})}-\tilde\delta_\star)}\right)},\\
 & \leq & 2\left( p^{(u+1)}\sqrt{1 + \frac{n \textsf{FP}}{\sigma^2\rho}}\right)^{\textsf{FP}},
 \end{eqnarray*}
where the second inequality uses the fact the eigenvalues of $L_\delta$ are all at least $1$, and (\ref{eq:norm:X}). In view of the above, we set
\begin{equation}\label{zeta:proof}
\zeta = \frac{\zeta_0^2}{8}\left( p^{(u+1)}\sqrt{1 + \frac{n\textsf{FP}}{\sigma^2\rho}}\right)^{\textsf{-2FP}},\end{equation}
which gives $\zeta\|f_0\|_{\pi,\infty}^2\leq \zeta_0^2/2$.
Therefore,  by Lemma \ref{lem:key} (applied with  $\|\cdot\|_\star = \|\cdot\|_{\pi,\infty}$), for all integer $N\geq 1$,   we have
\begin{equation}\label{eq1:proof:mix:lm}
\|\nu_0 K^N -\Pi(\cdot\vert z)\|_\tv^2 \leq  \|f_0\|_{\pi,\infty}^2 \left(1-\textsf{SpecGap}_\zeta(K)\right)^N+ \frac{\zeta_0^2}{2}.
 \end{equation}

\noindent\underline{\texttt{Lower bound on $\textsf{SpecGap}_\zeta(K)$}}.\;\; To proceed with (\ref{eq1:proof:mix:lm}) we need a lower bound on the approximate spectral gap.  We apply Theorem \ref{thm:mixing} with the obvious choices $\I=\Delta$, $\I_0= \mathcal{D}_{k}$, and $\cB_\delta=\rset^p$, and $m=+\infty$. 
For $z\in\e_k$, and $\zeta$ as in (\ref{zeta:proof}) we have 
\begin{equation*}\label{bound:max:1}
\frac{10}{\zeta}\left(1-\Pi(\mathcal{D}_{k}\vert z)\right) \leq \frac{320}{\zeta_0^2} \left( p^{(u+1)}\sqrt{1 + \frac{n\textsf{FP}}{\sigma^2\rho}}\right) ^{2\textsf{FP}} \frac{1}{p^{\frac{u}{2}(k+1)}} \leq 1,
\end{equation*}
provided  (\ref{cond:v}) holds. In other words we have $\Pi(\D_k\vert z)\geq 1-(\zeta/10)$ as required by Theorem \ref{thm:mixing}. It remains only to find $\kappa$. To do so, we   consider the follow graph on $\I_0$: we link $\delta^{(1)}$ and $\delta^{(2)}$ if $\delta^{(1)}\supseteq\delta^{(2)}$, or $\delta^{(2)}\supseteq\delta^{(1)}$, and $\|\delta^{(2)}-\delta^{(1)}|_0=1$. We need to  find $\kappa>0$ such that for all $\delta^{(1)},\\delta^{(2)}\in\mathcal{D}_{k}$, such that  if $\delta^{(1)}\supseteq\delta^{(2)}$, or $\delta^{(2)}\subseteq\delta^{(1)}$, and $\|\delta^{(2)}-\delta^{(1)}\|_0=1$ we have
\begin{equation}\label{cond:kappa}
 \int_{\rset^p}\min\left(\Pi(\theta\vert \delta^{(1)}, z),\Pi(\theta\vert \delta^{(2)}, z)\right) \rmd\theta\geq \kappa.\end{equation}
Suppose that $\delta^{(2)}\supseteq \delta^{(1)}$. Then
\[\frac{\Pi(\theta\vert \delta^{(2)}, z)}{\Pi(\theta\vert\delta^{(1)}, z)} \geq  \frac{\int_{\rset^p}e^{-\frac{1}{2\sigma^2}\|z-X\theta\|_2^2 -\frac{1}{2}\theta'D_{(\delta^{(1)})}^{-1}\theta}\rmd\theta}{\int_{\rset^p}e^{-\frac{1}{2\sigma^2}\|z-X\theta\|_2^2 -\frac{1}{2}\theta'D_{(\delta^{(2)})}^{-1}\theta}\rmd\theta} .\]
Using (\ref{eq1:proof:lm:step0}),   and (\ref{eq2:proof:lm:step0}) we have
\[\frac{\int_{\rset^p}e^{-\frac{1}{2\sigma^2}\|z-X\theta\|_2^2 -\frac{1}{2}\theta'D_{(\delta^{(1)})}^{-1}\theta}\rmd\theta}{\int_{\rset^p}e^{-\frac{1}{2\sigma^2}\|z-X\theta\|_2^2 -\frac{1}{2}\theta'D_{(\delta^{(2)})}^{-1}\theta}\rmd\theta}  \geq (\gamma\rho)  \sqrt{\frac{\det(L_{\delta^{(2)}})}{\det(L_{\delta^{(1)}})}} \frac{e^{\frac{1}{2\sigma^4}z'L^{-1}_{\delta^{(2)}} z}}{e^{\frac{1}{2\sigma^4}z'L^{-1}_{\delta^{(1)}} z}}.\]
We combine these inequalities with  (\ref{eq2:control:term1}) and (\ref{eq21:control:term1}) to get
\[\frac{\Pi(\theta\vert \delta^{(2)}, z)}{\Pi(\theta\vert \delta^{(1)}, z)} \geq (\gamma\rho) e^{-\frac{\tau(X_{(\delta^{(2)}-\delta^{(1)})}'L_{\delta^{(1)}}^{-1}z)^2}{2\sigma^2\left(1+ \tau n\varrho\right)}} .\]
For $j$ such that $\delta^{(2)}_j=1$, and $\delta^{(1)}_j=0$, we must have $\tilde\delta_{\star,j}=0$, since both $\delta^{(1)}$ and $\delta^{(2)}$ contain $\tilde\delta_\star$. Therefore, since $z = X\theta_\star + \sigma v$, where $\sigma = (z-X\theta_\star)/\sigma$, we have for $z\in\e_k$,
\begin{multline*}
|X_j' L_{\delta^{(1)}}^{-1}z|   =\left|\sigma X_j' L_{\delta^{(1)}}^{-1}v + \sum_{i:\;\tilde\delta_{\star,i}=0} \theta_{\star,i}X_j' L_{\delta^{(1)}}^{-1} X_i + \sum_{i:\;\tilde\delta_{\star,i}=1} \theta_{\star,i}X_j' L_{\delta^{(1)}}^{-1} X_i \right|\\
\leq 2\sigma\sqrt{(1+k)n\log(p)} + s_\star n\epsilon + \|\tilde\theta_\star\|_1\Cset(\tilde s_\star+ k)\\
\leq \sigma\sqrt{n\log(p)}\left(s_\star+ 2\sqrt{1+k} + \frac{\|\tilde\theta_\star\|_1\Cset(\tilde s_\star+ k)}{\sigma\sqrt{n\log(p)}}\right).
\end{multline*}
It follows that
\[\frac{\int_{\rset^p}e^{-\frac{1}{2\sigma^2}\|z-X\theta\|_2^2 -\frac{1}{2}\theta'D_{(\delta_1)}^{-1}\theta}\rmd\theta}{\int_{\rset^p}e^{-\frac{1}{2\sigma^2}\|z-X\theta\|_2^2 -\frac{1}{2}\theta'D_{(\delta_2)}^{-1}\theta}\rmd\theta}  \geq (\gamma\rho) p^{-\frac{1}{2\varrho}\left(s_\star+ 2\sqrt{1+k} + \frac{\|\tilde\theta_\star\|_1\Cset(\tilde s_\star+ k)}{\sigma\sqrt{n\log(p)}}\right)^2}.\]
Hence we can apply Theorem \ref{thm:mixing}  with 
\[\kappa = (\gamma\rho) p^{-\frac{1}{2\varrho}\left(s_\star+ 2\sqrt{1+k} + \frac{\|\tilde\theta_\star\|_1\Cset(\tilde s_\star +k)}{\sigma\sqrt{n\log(p)}}\right)^2}. \]
The diameter of the graph thus constructed is $2k$.  We conclude from the above and Theorem \ref{thm:mixing} that  for $z\in\e$,
 \begin{equation*}
 \textsf{SpecGap}_\zeta(K)\geq 
\frac{(\gamma\rho)}{4k} p^{-\frac{1}{2\varrho}\left(s_\star+ 2\sqrt{1+k} + \frac{\|\tilde\theta_\star\|_1\Cset(\tilde s_\star+ k)}{\sigma\sqrt{n\log(p)}}\right)^2}  \min_{\delta\in\mathcal{D}_{k}} \pi(\delta\vert z).
 \end{equation*}
Furthermore, for $\delta\in\mathcal{D}_k$, using (\ref{eq2:proof:lm:step0}) with $B=\{\delta\}$ and $\delta_0=\tilde\delta_\star$, together with  (\ref{eq2:control:term1}) and (\ref{eq21:control:term1}), we have
\begin{equation}\label{mino:pi}
\pi(\delta\vert z) \geq  \frac{1}{2} \frac{\pi(\delta\vert z)}{\pi(\tilde\delta_\star\vert z)}\geq \frac{1}{2} \frac{1}{p^{k(u+1)}}\left(1+ \frac{nk}{\sigma^2\rho}\right)^{-\frac{k}{2}}.\end{equation}
Hence 
 \begin{equation*}
 \textsf{SpecGap}_\zeta(K)\geq 
\frac{(\gamma\rho)}{8k} p^{-\frac{1}{2\varrho}\left(s_\star+ 2\sqrt{1+k} + \frac{\|\tilde\theta_\star\|_1\Cset(\tilde s_\star +k)}{\sigma\sqrt{n\log(p)}}\right)^2} p^{-k(u+1)} \left(1+ \frac{n k}{\sigma^2\rho}\right)^{-\frac{k}{2}}.
 \end{equation*}
It follows from (\ref{eq1:proof:mix:lm}) and the lower bound on $\textsf{SpecGap}_\zeta(K)$ above that for 
\begin{equation}\label{cond:it}
 N\geq \frac{A}{(\gamma\rho)} \log\left(\frac{1}{\zeta_0}\right)p^{\frac{1}{2\varrho}\left(s_\star+ 2\sqrt{1+k} + \frac{\|\tilde\theta_\star\|_1\Cset(\tilde s_\star +k)}{\sigma\sqrt{n\log(p)}}\right)^2} p^{k(u+1)} \left(1+ \frac{n k}{\sigma^2\rho}\right)^{\frac{k}{2}},\end{equation}
we have
\[\|\nu_0 K^N -\Pi(\cdot\vert z)\|_\tv \leq \zeta_0,\]
where $A$ is an absolute constant that does not depend on $p$ nor $\zeta_0$. This completes the proof.
\vspace{-0.6cm}
\begin{flushright}
$\square$
\end{flushright}

\medskip

\medskip

\appendix
\section{Some technical results}
We make use of the following standard Gaussian deviation bound.
%
\begin{lemma}\label{max:dev}
Let $Z\sim \textbf{N}(0,I_m)$, and $u_1,\ldots,u_N$ be vectors of $\rset^m$. Then for all $x\geq 0$,
\[\PP\left[\max_{1\leq j\leq N} \left|\pscal{u_j}{Z}\right|> \max_{1\leq j\leq N}\|u_j\|_2\sqrt{2(x+\log(N))}\right]\leq \frac{2}{e^x}.\]
\end{lemma}

The next result gives a bound on $\Cset_X$, and shows that H\ref{H:tech} holds with high probability in the case of a Gaussian ensemble. 

\begin{lemma}\label{lem:H:tech}
Suppose that $X\in\rset^{n\times p}$ is a random matrix with i.i.d. standard Normal entries. Given an integer $s$, and positive constants $\sigma,\gamma$ and $\rho$,  set 
\[\Cset(s) \eqdef \max_{\delta\in\Delta:\;\|\delta\|_0 \leq s} \;\; \max_{i\neq j,\;\delta_j=0}\;  \left| X_j'\left( I_n + \frac{1}{\sigma^2\rho}X_\delta X_\delta' +\frac{\gamma}{\sigma^2}X_{\delta^c}X_{\delta^c}'\right) X_i \right|.\]
Then there exist some universal  finite constants $c_0,a,A$ such that for $n\geq A s^2\log(p)$,  the following two statements hold with probability at least $1-\frac{a}{p}$: for $\gamma>0$ taken small enough and 
\begin{equation}\label{tech:lem:cond:gamma}
\sigma^2 s\rho \leq c_0 \sqrt{n\log(p)}, 
\end{equation}
it holds that
\begin{multline}\label{eq:lem:H:tech}
\Cset(s) \leq 2 c_0\sqrt{n\log(p)},\;\;\mbox{ and }\;\;\\
\;\;\min_{\delta:\;\|\delta\|_0\leq s}\;\;\inf\left\{ \frac{u'(X_{\delta^c}' L^{-1}_\delta X_{\delta^c}) u}{n\|u\|_2^2},\; u\in\rset^{p-s},\;0<\|\textsf{supp}(u)\|_0 \leq s\right\} \geq \frac{1}{32}.\end{multline}
\end{lemma}
\begin{proof}
For a matrix $M\in\rset^{n\times p}$ we set 
\[v(M,s) \eqdef \inf\left\{\frac{u'(M'M)u}{n\|u\|_2^2}\; u\neq 0, \|u\|_0 \leq s\right\},\]
and for $\kappa_0 = 1/64$ and $c_0=8$, we define 
\begin{multline*}
\e \eqdef \left\{M\in\rset^{n\times p}:\; v(M,s) \geq \kappa_0,\;\; \max_{1\leq j\leq p} \|M_j\|_2 \leq 2\sqrt{n},\;\right.\\
\left. \min_{1\leq j\leq p} \|M_j\|_2 \geq \sqrt{\frac{n}{2}},\;\;\mbox{ and }\;\; \max_{j\neq  k} |\pscal{M_j}{M_k}| \leq c_0\sqrt{n\log(p)}\right\}.
\end{multline*}
By Theorem 1 of \cite{raskutti:etal:10}, Lemma 1-(4.2) of \cite{laurent:massart:00}, and standard Gaussian deviation bounds, we can find universal constants $a,A$, such that for $n\geq A s\log(p)$, we have $\PP(X\notin \e) \leq \frac{a}{p}$. So to obtained the statement of the lemma,  it suffices to consider some arbitrary element $X\in\e$ and show that (\ref{eq:lem:H:tech}) holds.

Fix $\delta\in\Delta$ such that $\|\delta\|_0\leq s$. We set $M_\delta \eqdef I_n + \frac{1}{\sigma^2\rho}X_\delta X_\delta'$, so that $L_\delta = M_\delta +\frac{\gamma}{\sigma^2} X_{\delta^c}X_{\delta^c}'$.  The Woodbury identity gives
\begin{equation}\label{Ld}
L_\delta^{-1} = M_\delta^{-1} -\frac{\gamma}{\sigma^2}M_\delta^{-1} X_{\delta^c}\left(I_{\|\delta^c\|_0} +\frac{\gamma}{\sigma^2}X_{\delta^c}'M_\delta^{-1}X_{\delta^c}\right)^{-1}X_{\delta^c}' M_\delta^{-1}.\end{equation}
Hence, for any $j,k$,
\begin{equation}\label{jLk}
X_j'L_\delta^{-1} X_k  = X_j' M_\delta^{-1} X_k -\frac{\gamma}{\sigma^2} X_j' M_\delta^{-1} X_{\delta^c}\left(I_{\|\delta^c\|_0} + \frac{\gamma}{\sigma^2} X_{\delta^c}'M_\delta^{-1}X_{\delta^c}\right)^{-1} X_{\delta^c}'M_\delta^{-1} X_k.\end{equation}
If $C_1 = \max_{\ell} X_\ell' M_\delta^{-1} X_\ell$, and $C_0 = \max_{\ell\neq j,\;\delta_j=0} |X_j' M_\delta^{-1} X_\ell|$, then we deduce easily from (\ref{jLk}) that for all $j\neq k$ such that $\delta_j=0$,
\begin{equation}\label{jLk:2}
|X_j' L_\delta^{-1} X_k | \leq C_0 + \frac{\gamma}{\sigma^2} \left(C_1^2 + pC_0^2\right) .
\end{equation}
In order to proceed, we need to bound the term $X_jM_\delta^{-1} X_k$.  Easily, for $X\in\e$, we have
\[X_j'M_\delta^{-1}X_j \leq \|X_j\|_2^2\leq 4n.\]
Another application of the Woodbury identity gives
\begin{equation}\label{jMk:1}
M_\delta^{-1} = I_n - \frac{1}{\sigma^2\rho}X_\delta\left(I_{\|\delta\|_0} + \frac{1}{\sigma^2\rho}X_\delta'X_\delta\right)^{-1}X_\delta'.\end{equation}

If $X_\delta = U\Lambda V'$ is the singular value decomposition of $X_\delta$, with positive singular values $\lambda_1\geq \lambda_2 \ldots\geq \lambda_{\|\delta\|_0}$, and if $\mathcal{P}_{\perp}$ denotes the projector on the space orthogonal to the span of $X_\delta$, we have
\[M_\delta^{-1} = \mathcal{P}_\perp + \sum_{\ell=1}^{\|\delta\|_0} \frac{\sigma^2\rho}{\sigma^2\rho + \lambda_\ell^2} U_\ell U_\ell'.\]
We note that for $X\in\e$, $\lambda_{\|\delta\|_0}^2\geq \kappa_0 n$. Therefore,for $k\neq j$,  and using the above,
\[|X_j' M_\delta^{-1} X_k |\leq \left|\pscal{X_j}{\mathcal{P}_\perp(X_k)}\right| + \frac{\sigma^2 s\rho}{\kappa_0}\leq 2c_0\sqrt{n\log(p)},\]
provided that $\sigma^2 s\rho \leq c_0\kappa_0\sqrt{n\log(p)}$ as assumed in (\ref{tech:lem:cond:gamma}).
We combine this with (\ref{jLk:2}) to obtain that for $j\neq k$ such that $\delta_j=0$,
\begin{multline}\label{jLk:3}
|X_j' L_\delta^{-1} X_ k| \leq 3c_0\sqrt{n\log(p)}\left(1 + \frac{\gamma}{\sigma^2} pc_0\sqrt{n\log(p)}\right) + 16\frac{\gamma}{\sigma^2} n^2 \leq 8c_0\sqrt{n\log(p)},
\end{multline}
for $\gamma$ small enough. (\ref{jLk:3})  says that $\Cset_X\leq 8c_0\sqrt{n\log(p)}$, for $X\in\e$, as claimed.

For $j$ such that $\delta_j=0$, (\ref{jMk:1}) gives
\begin{eqnarray}\label{jMj}
X_j'M_\delta^{-1}X_j & = & \|X_j\|_2^2 - \frac{1}{\sigma^2\rho}X_j'X_\delta\left(I_{\|\delta\|_0} + \frac{1}{\sigma^2\rho}X_\delta'X_\delta\right)^{-1}X_\delta'X_j \nonumber\\
& \geq & \|X_j\|_2^2 -\frac{\frac{1}{\sigma^2\rho}\|X_\delta'X_j\|_2^2}{1 + \frac{n\kappa_0}{\sigma^2\rho}}\nonumber\\
& \geq & \|X_j\|_2^2 -\frac{\|X_\delta'X_j\|_2^2}{n\kappa_0}\nonumber\\
& \geq & \|X_j\|_2^2 - \frac{s c_0^2\log(p)}{\kappa_0},\nonumber\\
& \geq & \frac{n}{4},\end{eqnarray}
since $n\geq A s\log(p)$, and by taking $A$ large enough ($A\geq 4c_0^2/\kappa_0)$. Equation (\ref{Ld})  then yields
\begin{multline*}
X_j'L_\delta^{-1} X_j \geq X_j' M_\delta^{-1} X_j - \frac{\gamma}{\sigma^2} \|X_{\delta^c}'M_\delta^{-1}X_j\|_2^2 \\
= X_j' M_\delta^{-1} X_j - \frac{\gamma}{\sigma^2} \left[(X_j'M_\delta^{-1} X_j)^2 + \sum_{k:\;\delta_k=0,k\neq j} (X_j'M_\delta^{-1} X_k)^2\right].\end{multline*}
For $2\gamma\leq\sigma^2$, it follows that
\[X_j'L_\delta^{-1} X_j \geq\frac{n}{8}- \frac{\gamma}{\sigma^2}(p-\|\delta\|_0)\left(4 c_0^2 n\log(p)\right),\]
which together with (\ref{jLk:3}) and (\ref{tech:lem:cond:gamma})  implies that for any $u\in\rset^p$ such that $\delta^c\supseteq \textsf{supp}(u)$, and $\|\textsf{supp}(u)\|_0 \leq s$, we have
\[u'X_{\delta^c}'L_\delta^{-1} X_{\delta^c} u \geq \frac{n}{32} \|u\|_2^2,\]
as claimed.

\end{proof}

\section*{Acknowledgements}
I'm grateful to Joonha Park  for pointing out an error in an initial draft of the manuscript.


\bibliographystyle{ims}
\bibliography{biblio_graph,biblio_mcmc,biblio_optim}

\end{document}